\documentclass[preprint,12pt]{elsarticle}

\usepackage{color}
\usepackage{enumerate}

\usepackage{graphicx}
\usepackage{subcaption}

\usepackage{cancel}

\usepackage{amsmath}
\usepackage{scalerel}
\usepackage{amsfonts}
\usepackage{amssymb}
\usepackage{amsthm}
\usepackage{latexsym}

\usepackage{hyperref}
\theoremstyle{definition}
\newtheorem*{prop}{Property}
\theoremstyle{definition}
\newtheorem{definition}{Definition}[section]
\newtheorem{theorem}{Theorem}[section]
\newtheorem{property}{Property}[section]
\newtheorem{exercise}{Exercise}[section]
\newcommand{\Hom}{\mathop{}\mathopen{}{\mathrm{Hom}}}

\newcommand{\beval}[1]{\left[\kern-0.15em\left[ #1 \right]\kern-0.15em\right]}
\newcommand{\eval}[1]{[\kern-0.15em[ #1 ]\kern-0.15em]}

\DeclareMathOperator{\Ima}{Im}
\DeclareMathOperator{\U1}{\mathrm{U}\!\left(1\right)}
\DeclareMathOperator{\SU2}{\mathrm{SU}\!\left(2\right)}
\DeclareMathOperator{\DBc}{\overset{\lower0.5em\hbox{$\smash{\scriptscriptstyle\frown}$}}{ \star }}
\DeclareMathOperator{\Wc}{\overset{\lower0.5em\hbox{$\smash{\scriptscriptstyle\frown}$}}{ \wedge }}

\begin{document}
	
	\begin{frontmatter}
		
		
		
		\title{An extension of the $\mathrm{U}\!\left(1\right)$ BF theory,\\ Turaev-Viro invariant and Drinfeld center construction. \\ Part I: Quantum fields, quantum currents and Pontryagin duality}
		
		\author[add1]{Emil H{\o}ssjer}
		\ead{hossjer@imada.sdu.dk}
		\author[add2]{Philippe Mathieu}
		\ead{philippe.mathieu@math.uzh.ch}
		\author[add3]{Frank Thuillier}
		\ead{frank.thuillier@lapth.cnrs.fr}
		
		% \cortext[cor1]{Please address correspondence to Author2 or Author3}
		\address[add1]{Department of Mathematics and Computer Science (IMADA), University of Southern Denmark, Campusvej 55, DK-5230 Odense M, Denmark}
		\address[add2]{Institut f\"ur Mathematik, Universit\"at Z\"urich, Winterthurerstrasse 190, CH-8057 Z\"urich}
		\address[add3]{Université Grenoble Alpes, USMB, CNRS, LAPTh, F-74000 Annecy, France}
		
		\begin{abstract}
			
			In this first of a series of articles dedicated to natural extensions of the $\U1$ BF theory, abelian Turaev-Viro (TV) construction and corresponding Drinfeld center construction for any closed oriented smooth manifolds, we present the mathematical background that will be used.
			
		\end{abstract}
		
		% \begin{highlights}
			% \end{highlights}
		
		%		\begin{keyword}
			%			\textcolor{red}{OurKeyWord1 \sep OurKeyWord2 \sep OurKeyWord3 \sep OurKeyWord4} 
			%		\end{keyword}
		
	\end{frontmatter}
	
	
	
	\newpage
	
	\tableofcontents
	
	\newpage
	
	\section{Introduction}
	
	It is now well known that Deligne-Beilinson (DB) cohomology is a particularly efficient tool to describe and study the $\U1$ BF and Chern-Simons (CS) theories on a $3$-dimensional closed oriented smooth manifold $M$ \cite{GT2014,MT1}. Among the various consequences induced by the use of this cohomology, let us highlight the quantization of the coupling constant and charges carried by the loops defining the observables, as well as the non-perturbative determination of the partition function and of the expectation values of these observables. Within this framework, the fields are actually classes of $\U1$ gauge fields, the classical actions $k \int A \wedge dA$ and $k \int B \wedge dA$  being particular local contribution to the ``quantum" CS and BF actions. In each of these theories, the Lagrangian turns out to be a DB product of fields whose integral over $M$ defines the corresponding action as an $\mathbb{R}/\mathbb{Z}$-valued functional. At the end, it appears that the CS, resp. BF, partition function coincides, up to some reciprocity formula, with a Reshetikhin-Turaev (RT), resp. Turaev-Viro (TV), abelian invariant of $M$. Moreover, a discrete BF theory can be extracted from the TV construction which further underlines the relation between the abelian TV construction and the $\U1$ BF theory in $3$ dimensions \cite{MT2}. Eventually, it is possible to check that the Turaev-Virelizier theorem \cite{TV13} applies in this abelian context so that the TV invariant can also be obtained from an RT construction by using the Drinfeld center of some cyclic group \cite{MT3}. 
	
	In this series of articles, we want to present a natural extension of the above considerations for any closed oriented smooth manifold $M$. However, although Deligne-Beilinson cohomology group of $M$ is still the cornerstone of our extended $\U1$ BF theory, we will favor a more physical terminology by referring to DB classes as ``quantum fields". As in the $3$-dimensional case, the BF observables will be holonomies of quantum fields along cycles of $M$ so that cycles are particular examples of elements of the so-called Pontryagin dual of the group of quantum fields. Still in the $3$-dimensional case, it was  more or less explicitly claimed that the Pontryagin dual of the set of quantum fields can be identified with the set of quantum currents, i.e., collections whose components are local de Rham currrents instead of local forms, and thus that a $1$-cycle of $M$ canonically defines a quantum current. However, integration along a $1$-cycle is performed on quantum fields whereas, as we will see, a quantum current is naturally evaluated on a so-called ``dual quantum field". In fact, the canonical quantum current $\boldsymbol{Z}_{[p]}$ associated with a $p$-cycle $z$ of $M$ will be defined with the help of an extended DB product according to
	\begin{equation}
		\label{currentofacycle}
		\oint_{z} \boldsymbol{A}^{[p]} = (\boldsymbol{Z}_{[p]} \star \boldsymbol{A}^{[p]}) \eval{\boldsymbol{1}} \, ,
	\end{equation}
	whereas the dual quantum current $\boldsymbol{Z}_{\{p\}}$ associated with $z$ is simply defined by
	\begin{equation}
		\boldsymbol{Z}_{\{p\}} \eval{\boldsymbol{A}^{[p]}} = \oint_{z} \boldsymbol{A}^{[p]} \, .
	\end{equation}
	The main purpose of this first article is to give a precise meaning to all the above expressions, then to show how $\boldsymbol{Z}_{[p]}$ and $\boldsymbol{Z}_{\{p\}}$ are related and finally to extend the DB product to a pairing between quantum currents which can be seen as singular quantum fields. Accordingly, the $p$-cycles generating the observables of the generalized $\U1$ BF theory will be treated as quantum currents so that the shift procedure used to compute expectation values of observables will be mathematically consistent. Let us note that although it was already presented in \cite{BGST05}, we propose here a more rigorous and more detailed proof of equality \eqref{currentofacycle}. In particular, we will introduce various mathematical objects which all play a precise role in establishing \eqref{currentofacycle}.
	
	Beside this introduction, a set of mathematical reminders and a conclusion, this article is made of three main sections. 
	
	In Section \ref{section_quantum_fields} we introduce the notion of gauge fields, gauge field transformations and quantum fields. The exact sequences in which the set of quantum fields sit is then exhibited. Eventually, integration of quantum currents along cycles will be introduced as an $\mathbb{R}/\mathbb{Z}$-valued functional. Except for the terminology, there is nothing new in this section, so that the reader familiar with Deligne-Beilinson cohomology and its description in the \v{C}ech-de Rham complex can skip it.
	
	In Section \ref{section_quantum_currents}, we extend the construction of quantum fields to obtain quantum currents. These objects are to quantum fields what de Rham currents are to forms. We will show that the set of quantum currents also sits in some natural exact sequences, and an evaluation formula will be provided. The objects on which a quantum current is generically evaluated are dual quantum fields and the evaluation is an $\mathbb{R}/\mathbb{Z}$-valued functional. Then, we show how to naturally extend the DB product into a pairing between quantum $p$-fields and quantum $p$-currents. This extension relies on a pairing between quantum fields and dual quantum fields in a way that is very similar to how the cap and the cup products are related to each other.
	
	In Section \ref{section_pontryagin_duality}, we show how Pontryagin duality and quantum currents are related to each other. This is achieved by using dual quantum currents. Just like a quantum current can be represented by a gauge current, a dual quantum current can be represented by a dual gauge current which is a collection of local de Rham currents. However, the components of a dual gauge current must have compact supports and fulfill a descent in \v{C}ech homology whereas the components of a gauge current have no support requirement and fulfill a descent in \v{C}ech cohomology. Then, we show that dual gauge currents provide representatives of the elements of the Pontryagin dual of the set of quantum fields and explain how the use of dual gauge currents yields equality \eqref{currentofacycle}. Finally, we propose an extension of the DB product to a product between quantum $q$-currents and quantum $p$-currents, with $p+q+1=n$. While this product is in general ill-defined, just like the exterior product of two de Rham currents is, sometimes, as with de Rham currents, this extended DB product is well-defined.
	
	Let us summarize our main results into the following property:
	
	\begin{prop}
		Let $M$ be a connected closed oriented smooth manifold of dimension $n$. We denote by $H_D^p(M)$ and $H^D_p(M)$ the $\mathbb{Z}$-modules of quantum fields and quantum currents of $M$, respectively.
		\begin{enumerate}[1)]
			\item On the one hand, quantum $p$-currents are evaluated on dual quantum $p$-fields, the $\mathbb{Z}$-module of which is denoted by $D_H^p(M)$. On the other hand, dual quantum $p$-currents are evaluated on quantum $p$-fields so that $H_D^p(M)^\star$, the Pontryagin dual of $H_D^p(M)$, can be canonically identified with $D_H^p(M)$. The representatives of the elements of $D_H^p(M)$ hence provide representatives of the elements of $H_D^p(M)^\star$. Then, we have the following sequence of isomorphisms
			\begin{equation}
				H^D_p(M) \xrightarrow{\; \; \mu \; \;} D_p^H(M) \xrightarrow{eval} H_D^p(M)^\star \, ,
			\end{equation}
			where $\mu$ refers to some partition of unity subordinate to a good cover of $M$ and $eval$ refers to the evaluation of dual quantum $p$-currents on quantum $p$-fields. The composition of these two isomorphisms induces a canonical isomorphism between $H^D_p(M)$ and $H_D^p(M)^\star$.
			
			\item If $z$ is a $p$-cycle of $M$, then there exists a unique dual quantum $p$-current $\boldsymbol{Z}_{\{p\}}$ and a unique quantum $p$-current $\boldsymbol{Z}_{[p]}$ such that
			\begin{equation}
				\boldsymbol{Z}_{\{p\}} \eval{\boldsymbol{A}^{[p]}} = \oint_{z} \boldsymbol{A}^{[p]} = \oint_M \boldsymbol{Z}_{[p]} \star \boldsymbol{A}^{[p]} \, , \nonumber
			\end{equation}
			for any quantum $p$-field $\boldsymbol{A}^{[p]}$. The above equalities are written in $\mathbb{R}/\mathbb{Z}$.
			
			\item In some cases, it is possible to give a meaning to the DB product of a quantum $p$-current and a quantum $q$-current. For instance, when the quantum currents derive from quantum fields their DB product is nothing but the usual DB product.
		\end{enumerate}
	\end{prop}	
	
	All along this article, $M$ denotes a $n$-dimensional connected closed oriented smooth manifold. On another note, we will often deal with collections of local forms or currents. Conventionally, the first index of such a collection will refer to its \v{C}ech degree and the second to its de Rham degree. Moreover, a lower index (for \v{C}ech or de Rham) will refer to homology and an upper index to cohomology.
	
	\section{Basic mathematical reminders}
	\label{basics}
	
	As they will play an important role in this article, let us recall some basic facts regarding homology and cohomology:
	
	\begin{itemize}
		\item[$\bullet$] Since Poincaré lemma is at the heart of the various constructions we are going to study, let us recall that 
		\begin{equation}
			\left\{ \begin{gathered}
				H^p(\mathbb{R}^n) = 0 \quad \mbox{ for } p \in \{1 , \dots , n\} \, , \hfill \\
				H^0(\mathbb{R}^n) = \mathbb{Z} \, , \hfill \\ 
			\end{gathered}  \right.
		\end{equation}
		whereas for the cohomology with compact support
		\begin{equation}
			\left\{ \begin{gathered}
				H_c^p(\mathbb{R}^n) = 0 \quad \mbox{ for } p \in \{0 , \dots , n-1\} \, , \\
				H_c^n(\mathbb{R}^n) = \mathbb{Z}  \, . \hfill \\ 
			\end{gathered}  \right.
		\end{equation}
		Thus, any closed $p$-form of $\mathbb{R}^n$ is exact when $p \in \{1 , \dots , n\}$ and any closed $p$-form with compact support of $\mathbb{R}^n$ is exact (with respect to forms with compact support) when $p \in \{0 , \dots , n-1\}$ with the particular case that a constant function with compact support in $\mathbb{R}^n$ is necessarily zero.
		
		\item[$\bullet$] Let $H_p(M)$ be the $p$-th (singular, simplicial or \v{C}ech, depending on the context) homology group of $M$. As a finitely generated abelian group, $H_p(M)$ is the direct sum of a free sector and a torsion sector, respectively denoted by $F_p$ and $T_p$ and such that
		\begin{equation}
			\label{decomphomo}
			F_p \simeq \mathbb{Z}^{b_p} \quad\mbox{ and }\quad T_p \simeq\mathbb{Z}_{\zeta^1_p} \oplus \hdots \oplus \mathbb{Z}_{\zeta^{t_p}_p}\, ,
		\end{equation}
		with $b_p$ the $p$-th Betti number of $M$ and  $\zeta^i_p$ dividing $\zeta^{i+1}_p$ for $i \in \left\lbrace 0, 1, \hdots\right.$ $\left. \hdots, t_p - 1\right\rbrace$. The $p$-cohomology group of $M$ is denoted $H^p(M)$ with $H^p\left(M\right) = F^p \oplus T^p$, the real cohomology group then being $H^p\left(M,\mathbb{R}\right) := H^p\left(M\right) \otimes_{\mathbb{R}} \mathbb{R} = F^p \otimes_{\mathbb{R}} \mathbb{R}$.
		
		\item[$\bullet$] The Universal Coefficient Theorem \cite{BT82} implies the following isomorphisms
		\begin{enumerate}[(1)]
			\item $H^p \left(M\right) \simeq F_p \oplus T_{p-1}$,
			\item $H^p \left(M,\mathbb{R}/\mathbb{Z}\right) \simeq \left( \mathbb{R}/\mathbb{Z} \right)^{b_p} \oplus T_p$, \\
			\item $H^p \left(M,\mathbb{Z}_N\right) \simeq \left(\mathbb{Z}_N\right)^{b_p} \oplus \bigoplus_{i=1}^{t_p} \mathbb{Z}_{\mathrm{gcd}\left(\zeta^i_p,N\right)} \oplus \bigoplus_{j=1}^{t_{p-1}} \mathbb{Z}_{\mathrm{gcd}\left(\zeta^j_{p - 1},N\right)}$.
		\end{enumerate}
		
		\item[$\bullet$] Poincaré duality reads
		\begin{equation}
			\label{PoincarDual}
			H^{p+1}\left(M\right) \simeq H_{n-p-1}\left(M\right)\, ,
		\end{equation}
		and when combined with the first isomorphism coming from the Universal Coefficient Theorem recalled above, it implies
		\begin{equation}
			T_{n-p-1} \simeq T_p\, .
		\end{equation}
		
		\item[$\bullet$] Let $\Omega^p(M)$ denote the space of smooth $p$-forms over $M$. The subspace of closed $p$-forms is denoted by $\Omega^p_\circ\left(M\right)$ and the subspace of closed $p$-forms with integral periods by $\Omega^p_\mathbb{Z}(M)$. We have the sequence of inclusions:
		\begin{equation}
			\label{inclusions1}
			\Omega^p_\mathbb{Z}(M) \subset \Omega^p_\circ(M) \subset \Omega^p(M)\, .
		\end{equation} 
		Furthermore, with respect to the exterior derivative $d$, $\Omega^p_\mathbb{Z}(M)$ gives rise to a cohomology group which is isomorphic to $F^p$, whereas $\Omega^p_\circ(M)$ gives rise to the de Rham $p$-th cohomology group of $M$ which is isomorphic to $H^p(M,\mathbb{R}) = H^p(M) \otimes_{\mathbb{R}} \mathbb{R} = F^p \otimes_{\mathbb{R}} \mathbb{R}$ \cite{dR55,We52}.
		
		\item[$\bullet$] Let $\Omega_p(M)$ denote the topological dual of $\Omega^p(M)$, i.e., the space of continuous linear functionals on $\Omega^p(M)$, the elements of which are called de Rham $p$-currents \cite{dR55}. The integer $p$ is usually referred to as the dimension of the de Rham current and the integer $(n-p)$ as its degree. We endow $\Omega_p(M)$ with the differential $d^\dagger: \Omega_p(M) \to \Omega_{p-1}(M)$ which is dual to $d$ according to
		\begin{equation}
			\label{dualdeRhamoperators}
			(d^\dagger J) [\omega] = J [d \omega] \, ,
		\end{equation}
		for any $J \in \Omega_p(M)$ and any $\omega \in \Omega^{p-1}(M)$. The subspace of $d^\dagger$-closed de Rham $p$-currents is denoted by $\Omega_p^\circ(M)$, while the $\mathbb{Z}$-module of $d^\dagger$-closed de Rham $p$-currents which are $\mathbb{Z}$-valued on $\Omega^p_\mathbb{Z}(M)$ is denoted by $\Omega_p^\mathbb{Z}(M)$. From the point of view of the dimension, $d^\dagger$ is a homology operator whereas it is a cohomology operator from the point of view of the degree. However, for later convenience and because of relation \eqref{ddaggervsd} below, we refer to $d^\dagger$ as a cohomology operator. It is a well-known result that the cohomology of currents is the same as the cohomology of forms \cite{dR55} so that the cohomology of $(\Omega_p(M),d^\dagger)$ is isomorphic to $H_p(M,\mathbb{R}) \simeq H^{n-p}(M,\mathbb{R}) \simeq H^p(M,\mathbb{R})$. The cohomology associated with $\Omega_p^\mathbb{Z}(M)$ is isomorphic to $F_p \simeq F^{n-p} \simeq F^p$. Last but not least, if $U$ is an open subset of $M$, the space of de Rham $p$-currents of $U$ is the topological dual of $\Omega_c^p(U)$, the space of $p$-forms with compact support in $U$. Moreover, if $U$ is contractible, then the Poincaré lemma can be applied in $U$ so that any closed de Rham $p$-current of $U$ is exact when $0 < p < n$.

		\item[$\bullet$] Through integration over $p$-forms, a singular $p$-chain of $M$ defines an element of $\Omega_p(M)$. From this integration point of view, two $p$-chains which have the same integral over any $p$-form are said to be equivalent and the corresponding equivalence class is sometimes called a de Rham $p$-chain \cite{BGST05}. As in this article integration plays a central role, we will refer to a de Rham $p$-chain of $M$ simply as a $p$-chain, the abelian free group of (de Rham) $p$-chains of $M$ hence being denoted by $C_p\left(M\right)$ and the subgroup of (de Rham) $p$-cycles by $Z_p\left(M\right)$. Under the identification of a $p$-chain with the de Rham current it defines, we end with the canonical inclusions
		\begin{equation}
			\label{cyclesascurrents}
			C_p(M) \subset \Omega_p(M)\, ,  \quad Z_p(M) \subset \Omega_p^\circ(M) \, .
		\end{equation}
		
		\item[$\bullet$] By combining the exterior product with integration over $M$ according to
		\begin{equation}
			\label{convention}
			\omega [\alpha] = \oint_M \omega \wedge \alpha \, .
		\end{equation}
		we obtain the injection:
		\begin{equation}
			\label{forsmascurrents}
			\Omega^{n-p}\left(M\right) \overset{i^\diamond}{\hookrightarrow} \Omega_p\left(M\right) \, .
		\end{equation}
		which is often written as an inclusion. The elements of $\Ima i^\diamond$ are usually referred to as regular de Rham $p$-currents. In fact, a de Rham $p$-current can be represented by a $(n-p)$-form with distribution coefficients \cite{dR55}. With respect to this representation, we have 
		\begin{equation}
			\label{ddaggervsd}
			d^\dagger = (-1)^{n-p-1} d \, .
		\end{equation}
		
		Of course, $i^\diamond(\Omega^{n-p}_\circ(M)) \subset \Omega_p^\circ(M)$ and $i^\diamond(\Omega^{n-p}_\mathbb{Z}(M)) \subset \Omega_p^\mathbb{Z}(M)$. Moreover, $\Ima i^\diamond$ is dense in $\Omega_p\left(M\right)$ with respect to the weak topology \cite{dR55}. Finally, the exterior product of a de Rham $l$-current $J$ with a $k$-form $\omega$ is the $(l-k)$-current $J \wedge \omega$ defined by
		\begin{equation}
			\label{productcurrentform}
			(J \wedge \omega) [\alpha] = J [\omega \wedge \alpha] \, ,
		\end{equation}
		with of course $0 \leq l - k \leq n$.
	\end{itemize}
	
	\section{Quantum fields}
	\label{section_quantum_fields}
	
	Let us provide $M$ with a finite good cover $\mathcal{U}_M = \left\lbrace U_\alpha \right\rbrace_{\alpha \in I}$, i.e., a finite collection of open subsets $U_\alpha \subset M$ such that
	\begin{equation}
		\label{finitegoodcover}
		\left\{ \begin{gathered}
			\bigcup\limits_{\alpha \in I} {U_\alpha } = M \, , \hfill \\
			U_{\alpha_0 \hdots \alpha_k} = U_{\alpha_0} \cap \hdots \cap U_{\alpha_k} \mathop \simeq \limits^{\mathrm{diff.}} \mathbb{R}^n \mbox{ or } \varnothing \, . \hfill \\ 
		\end{gathered}  \right.
	\end{equation}
	A collection $a^{(p+1,-1)}$ of integers, one in each non-empty $U_{\alpha_0 \hdots \alpha_{p+1}}$, is called a \v{C}ech $(p+1)$-cochain of $\mathcal{U}_M$. If the $a_{\alpha_0 \hdots \alpha_{p+1}}^{(p+1,-1)}$ are just real numbers then $a^{(p+1,-1)}$ is called a real \v{C}ech $(p+1)$-cochain of $\mathcal{U}_M$. We set
	\begin{equation}
		\label{Cechcocycle}
		\left( \delta a^{(p+1,-1)} \right)_{\alpha_0 \hdots \alpha_{p+2}} = \sum\limits_{i=0}^{p+2} (-1)^i a_{\alpha_0 \hdots {\check{\alpha}}_i \hdots \alpha_{p+2}}^{(p+1,-1)} \, ,
	\end{equation}
	with ${\check{\alpha}}_i$ denoting the omission of $\alpha_i$. By construction, the operator $\delta$ is nilpotent ($\delta^2 = 0$) and hence a cohomology operator. Accordingly, if $\delta a^{(p+1,-1)} = 0$ we say that the (real) \v{C}ech cochain $a^{(p+1,-1)}$ is a (real) \v{C}ech cocycle. The (real) \v{C}ech cohomology of $M$ is then inferred. For $k\in \{0 , \dots , n\}$, let $A^{(k,p-k)}$ be a collection of local $(p-k)$-forms, one in each non-empty intersection $U_{\alpha_0 \hdots \alpha_k}$. We say that $A^{(k,p-k)}$ is a \v{C}ech-de Rham $(k,p-k)$-cochain of $\mathcal{U}_M$, and we naturally set
	\begin{equation}
		\left( \delta A^{(k,p-k)} \right)_{\alpha_0 \hdots \alpha_{k + 1}} = \sum\limits_{i=0}^{k+1} (-1)^i A_{\alpha_0 \hdots {\check{\alpha}}_i \hdots \alpha_{k+1}}^{(k,p-k)}  \, .
	\end{equation}
	As they are contractible, Poincaré lemma can be applied in any of the intersections $U_{\alpha_0 \hdots \alpha_k}$ which is the keystone of the \v{C}ech-de Rham procedure used to show that de Rham cohomology is the same as \v{C}ech real cohomology \cite{BT82,We52} and on which gauge fields fundamentally rely. 
	
	\subsection{From gauge fields to quantum fields}
	
	We can naively consider gauge fields as an attempt to classify the various \v{C}ech-de Rham descents of a given closed form with integral periods of $M$. This attempt starts with the following definition.
	
	\begin{definition} \label{gaugefield}
		Let $A^{[p]} = \left( A^{(0,p)} , A^{(1,p-1)} , \hdots , A^{(p,0)} , a^{(p+1,-1)} \right)$ be a $(p+2)$-tuple where $A^{(k,p-k)}$ is a \v{C}ech-de Rham $(k,p-k)$-cochain of $\mathcal{U}_M$ and $a^{\left( p+1,- 1 \right)}$ a \v{C}ech $(p+1)$-cochain of $\mathcal{U}_M$. We say that $A^{[p]}$ is a \textbf{gauge $p$-field} if
		\begin{equation}
			\label{descent}
			\left\{\begin{array}{l}
				\left( \delta A^{(k,p-k)} \right)_{\alpha_0 \hdots \alpha_{k + 1}}  = d A_{\alpha_0 \hdots \alpha_{k+1}}^{(k+1,p-k-1)}\, , \\
				{\left( \delta A^{(p,0)} \right)_{{\alpha _0} \hdots {\alpha _{p + 1}}}}  = d_{-1} a_{\alpha_0 \hdots \alpha_{p+1}}^{(p+1,-1)}\, ,
			\end{array} \right.
		\end{equation}
		with $k \in \{0 , \hdots , p\}$, $d_{-1}$ the canonical injection of numbers into (constant) functions. The above constraints are referred to as \textbf{the descent equations} of $A^{[p]}$. 
	\end{definition}
	
	The last equation of the descent equations of $A^{[p]}$ implies that $\delta a^{(p+1,-1)} = 0$ and hence that $a^{(p+1,-1)}$ is a \v{C}ech $(p+1)$-cocycle of $\mathcal{U}_M$. In fact, a gauge $p$-field is a Deligne-Beilinson $p$-cocycle of $\mathcal{U}_M$. However, the cohomological nature of the construction will be of little interest to us, hence our choice of a more ``physical'' terminology.
	
	There is an equivalence relation for gauge $p$-fields induced by the identification of the-called gauge transformations.
	
	\begin{definition} \label{gaugefieldtransfo}
		A \textbf{gauge $p$-field transformation} is a gauge $p$-field of the form
		\begin{equation}
			\label{gaugetransfo}
			\left( d G^{\left( 0,p-1 \right)} , \delta G^{\left( 0,p-1 \right)} + d G^{\left( 1,p-2 \right)}, \hdots ,\delta G^{\left( p-1,0 \right)} + d_{-1} g^{\left( p,-1 \right)} , \delta g^{\left( p,-1 \right)} \right)\, ,
		\end{equation}
		where $g^{\left( p,-1 \right)}$ is a \v{C}ech $p$-cochain of $\mathcal{U}_M$. Two gauge $p$-fields which differ by a gauge field transformation are said to be gauge equivalent, or simply equivalent. The equivalence class of a gauge $p$-field is called a \textbf{quantum $p$-field}, the quantum $p$-field of a gauge $p$-field $A^{[p]}$ being denoted by $\boldsymbol{A}^{[p]}$. The set of quantum $p$-fields of $M$, $H_D^p\left(M\right)$, is an additive abelian group or, equivalently, a $\mathbb{Z}$-module. 
	\end{definition}
	
	\subsection{Exact sequences associated with quantum fields}
	
	The group $H_D^p\left(M\right)$ has a very nice description that we recall now.
	
	\begin{theorem}
		The space $H_D^p\left(M\right)$ sits into the short exact sequence
		\begin{equation}
			\label{short2}
			0 \to \frac{\Omega^p(M)}{\Omega_\mathbb{Z}^p(M)} \xrightarrow{\bar{\delta}} H_D^p(M) \xrightarrow{cl} H^{p + 1}(M) \to 0 \, .
		\end{equation}
	\end{theorem}	
	
	\noindent As a detailed and direct proof of this theorem can be found for instance in \cite{Br93,BGST05} but also indirectly in \cite{CS73,HLZ}, let us just define the morphisms $cl$ and $\bar{\delta}$ appearing in the above sequence. By definition, the last component of a gauge field $A^{[p]}$ is a \v{C}ech cocycle $a^{(p+1,-1)}$ and that of a gauge field transformation a \v{C}ech coboundary. Hence, $cl$ associates with the quantum field defined by $\boldsymbol{A}^{[p]}$ the \v{C}ech cohomology class of $a^{(p+1,-1)}$. Next, for any $\chi^{(p)} \in \Omega^p(M)$ we consider the collection $\delta_{-1} \chi^{(p)}$ of its restrictions to the open sets $U_\alpha$ of $\mathcal{U}_M$
	\begin{equation}
		\left( \delta _{-1} \chi^{\left( p \right)} \right)_\alpha  = \left. \chi ^{\left( p \right)} \right|_{U_\alpha}\, .
	\end{equation}
	Then, we set
	\begin{equation}
		\label{classical}
		\chi^{[p]} = \left( \delta_{-1} \chi^{(p)} , 0 , \hdots ,0 \right)\, .
	\end{equation}
	This $(p+2)$-tuple is obviously a gauge $p$-field, the quantum field of which is denoted by $\boldsymbol{\chi}^{[p]}$. The association $\chi^{(p)} \mapsto \boldsymbol{\chi}^{[p]}$ gives rise to a map $\Omega^p(M) \to H_D^p(M)$. Now, if $\chi^{(p)} \in \Omega^p_\mathbb{Z}(M)$ then it is locally exact and hence $\delta_{-1} \chi^{(p)} = d \chi^{(0,p-1)}$ for some collection $\chi^{(0,p-1)}$ of local $(p-1)$-forms. By adding to $\chi^{[p]}$ the gauge field transformation $\left( - d \chi^{(0,p-1)} , - \delta \chi^{(0,p-1)} , 0 , \hdots , 0\right)$ we obtain the equivalent gauge field $\tilde{\chi}^{[p]} = \left(0 , \delta \chi^{(0,p-1)} , 0 , \hdots ,0 \right)$. Since $\tilde{\chi}^{[p]}$ fulfills the descent equations we deduce that $d \delta \chi^{(0,p-1)} = 0$ and hence, by Poincaré lemma, $\delta \chi^{(0,p-1)} = d \chi^{(1,p-2)}$ for some collection $\chi^{(1,p-2)}$ of local $(p-2)$-forms. By going on this way, we eventually get a gauge field $\breve{\chi}^{[p]} = \left(0 , \hdots , \delta \chi^{(p-1,0)} , 0 \right)$, which is equivalent to $\chi^{[p]}$. Here again, the descent equations imply that $d \delta \chi^{(p-1,0)} = 0$. However, as the elements of $\chi^{(p-1,0)}$ are local functions, this last constraint means that $\delta \chi^{(p-1,0)} = d_{-1} u^{(p-1,-1)}$ with $u^{(p-1,-1)}$ a collection of numbers, and we have  $\breve{\chi}^{[p]} = \left(0 , \hdots , d_{-1} u^{(p-1,-1)} , 0 \right)$. The last descent equation implies $\delta u^{(p-1,-1)} = 0$. However, since $\chi^{(p)} \in \Omega^p_\mathbb{Z}(M)$, what we did is nothing but a \v{C}ech-de Rham descent of $\chi^{(p)}$ and as de Rham cohomology of $\Omega^p_\mathbb{Z}(M)$ is isomorphic to the free sector of the \v{C}ech cohomology of $M$, we conclude that this descent can always be chosen in such a way that $u^{(p-1,-1)}$ is a(n integral) \v{C}ech cocycle. Thus, $\breve{\chi}^{[p]}$ can be written as $\left(0 , \hdots , d_{-1} u^{(p-1,-1)} , 0 = \delta u^{(p-1,-1)} \right)$ which shows that it is actually a gauge field transformation, and hence so does $\chi^{(p)}$. Consequently, all the elements of $\Omega^p_\mathbb{Z}(M)$ are in the kernel of the map $\Omega^p(M) \to H_D^p(M)$ so that this map can be consistently reduced to a map $\bar{\delta}: \Omega^p(M)/\Omega_\mathbb{Z}^p(M) \to H_D^p(M)$.
	
	A $p$-form of $M$ will also be called a \textbf{classical $p$-field}. If $\bar{\chi}^{(p)}$ denotes the class of $\chi^{(p)}$ in $\Omega^p(M)/\Omega_\mathbb{Z}^p(M)$, we have $\boldsymbol{\chi}^{[p]} = \bar{\delta}(\bar{\chi}^{(p)})$. The gauge $p$-field $\chi^{[p]}$ as defined in \eqref{classical} is then a representative of $\boldsymbol{\chi}^{[p]}$.
	
	Let us finally stress out that exact sequence \eqref{short2} implies that quantum fields rely on $M$, and not on the good cover $\mathcal{U}_M$ as gauge fields do.  
	
	\vspace{0.5cm}
	
	The curvature of a gauge $p$-field $A^{[p]}$ is the collection $F(A)_\alpha = d A_\alpha^{(0,p)}$. Together with the nilpotency property $d^2 = 0$, the descent equations of $A^{[p]}$ yield the following sequence of identities
	\begin{equation}
		F(A)_\beta - F(A)_\alpha = d A_\beta^{(0,p)} - d A_\alpha^{(0,p)} = d \left( d A_{\alpha \beta}^{(1,p-1)} \right) = 0\, ,
	\end{equation}
	which shows that the collection $F(A)_\alpha$ defines a $(p+1)$-form $F(A)$. Hence, a curvature is a classical $(p+1)$-field. This classical $(p+1)$-field is closed since its local representatives are exact. The descent equations of $A^{[p]}$ provide \v{C}ech-de Rham descent equations of the closed form $F(A)$ which end by the \v{C}ech cocycle $a^{(p+1,-1)}$, thus implying that $F(A)$ has integral periods \cite{We52}. So we have
	\begin{equation}
		F(A) \in \Omega^{p+1}_\mathbb{Z}(M)\, .
	\end{equation}
	It is obvious that the curvature of a gauge field transformation is zero. Equivalent gauge fields hence have the same curvature and we can talk about the \textbf{curvature of a quantum $p$-field}. A quantum field whose curvature is zero is said to be \textbf{flat}. From all this, a second exact sequence into which $H_D^p(M)$ sits can be exhibited \cite{HLZ}. 
	
	\begin{theorem}
		The space $H_D^p(M)$ sits into the short exact sequence
		\begin{equation}
			\label{short3}
			0 \to {H^p}(M,\mathbb{R}/\mathbb{Z}) \xrightarrow{i} H_D^p(M) \xrightarrow{\bar{d}} \Omega^{p + 1}_\mathbb{Z}(M) \to 0 \, .
		\end{equation}
	\end{theorem}
	
	\noindent In the above sequence, the morphism $\bar{d}$ associates with a quantum field its curvature, $\bar{d} \boldsymbol{A}^{[p]} = F(\boldsymbol{A}^{[p]})$. Now, let $\boldsymbol{r}$ be an element of ${H^p}(M,\mathbb{R}/\mathbb{Z})$ and let $r^{(p,-1)}$ be a real \v{C}ech cochain of $\mathcal{U}_M$ which represents $\boldsymbol{r}$. Then, $\delta r^{(p,-1)}$ is an integer-valued \v{C}ech $(p+1)$-cocycle of $\mathcal{U}_M$ and the collection $R^{[p]} = \left( 0 , \hdots , d_{-1} r^{(p,-1)} , \delta r^{(p,-1)} \right)$ is a gauge $p$-field. If $\tilde{r}^{(p,-1)}$ is another \v{C}ech representative of $\boldsymbol{r}$ then $\tilde{r}^{(p,-1)} = r^{(p,-1)} + m^{(p,-1)}$, for some \v{C}ech cochain $m^{(p,-1)}$, so that the gauge field $\tilde{R}^{[p]} = \left(0 , \hdots , d_{-1} \tilde{r}^{(p,-1)} , \delta \tilde{r}^{(p,-1)} \right)$ is gauge equivalent to $R^{[p]}$ since $\tilde{R}^{[p]} = R^{[p]} + \left(0 , \hdots , d_{-1} m^{(p,-1)} , \delta m^{(p,-1)} \right)$. Thus, by associating to $\boldsymbol{r}$ the quantum $p$-field $\boldsymbol{R}^{[p]}$, we define the map $i: {H^p}(M,\mathbb{R}/\mathbb{Z}) \to H_D^p(M)$ appearing in \eqref{short3}.
	
	Let us make two remarks concerning exact sequence \eqref{short3}. Firstly, from this sequence, we deduce that the cohomology group ${H^p}(M,\mathbb{R}/\mathbb{Z})$, or rather its image by $i$, can be identified with the set of flat quantum $p$-fields since $\Ima i = \ker \bar{d}$. Secondly, the group $\Omega^{p+1}_\mathbb{Z}(M)$ can be considered independently of $H_D^p(M)$, so that we can talk about curvatures without referring to quantum fields. To our knowledge, this is specific to the abelian case.
	
	Exact sequences \eqref{short2} and \eqref{short3} describe the $\mathbb{Z}$-module $H_D^p(M)$ into two different but equivalent ways, see Figure \ref{Space_Fields}. In the first description, $H_D^p(M)$ is mapped by $cl$ into the discrete ``base space" $H^{p+1}(M)$, while in the second description it is mapped by $\bar{d}$ into the multiply connected base space $\Omega^{p + 1}_\mathbb{Z}(M)$ which is endowed with the usual topology \cite{dR55}. The spaces $\Omega^p(M)/\Omega_\mathbb{Z}^p(M)$ and ${H^p}(M , \mathbb{R}/\mathbb{Z})$ are then groups acting on the corresponding ``fibers" of $H_D^p(M)$. Thus, if we choose an origin on a fiber of \eqref{short2}, resp. \eqref{short3}, then any other point of this fiber is reached by translating the chosen origin with the image by $\bar{\delta}$, resp. $i$, of an element of $\Omega^p(M)/\Omega_\mathbb{Z}^p(M)$, resp. ${H^p}( M,\mathbb{R}/\mathbb{Z})$. Thus, the group $\Omega^p(M)/\Omega_\mathbb{Z}^p(M)$ and ${H^p}( M,\mathbb{R}/\mathbb{Z})$, can be referred to as the group of \textbf{quantum $p$-field translations} of \eqref{short2} and \eqref{short3}, respectively, these two exact sequences being themselves referred to as ``affine fibrations" of $H_D^p(M)$.
	
	\begin{figure}
		\centering
		\begin{subfigure}[b]{0.8\textwidth}
			\centering
			\includegraphics[width=\textwidth]{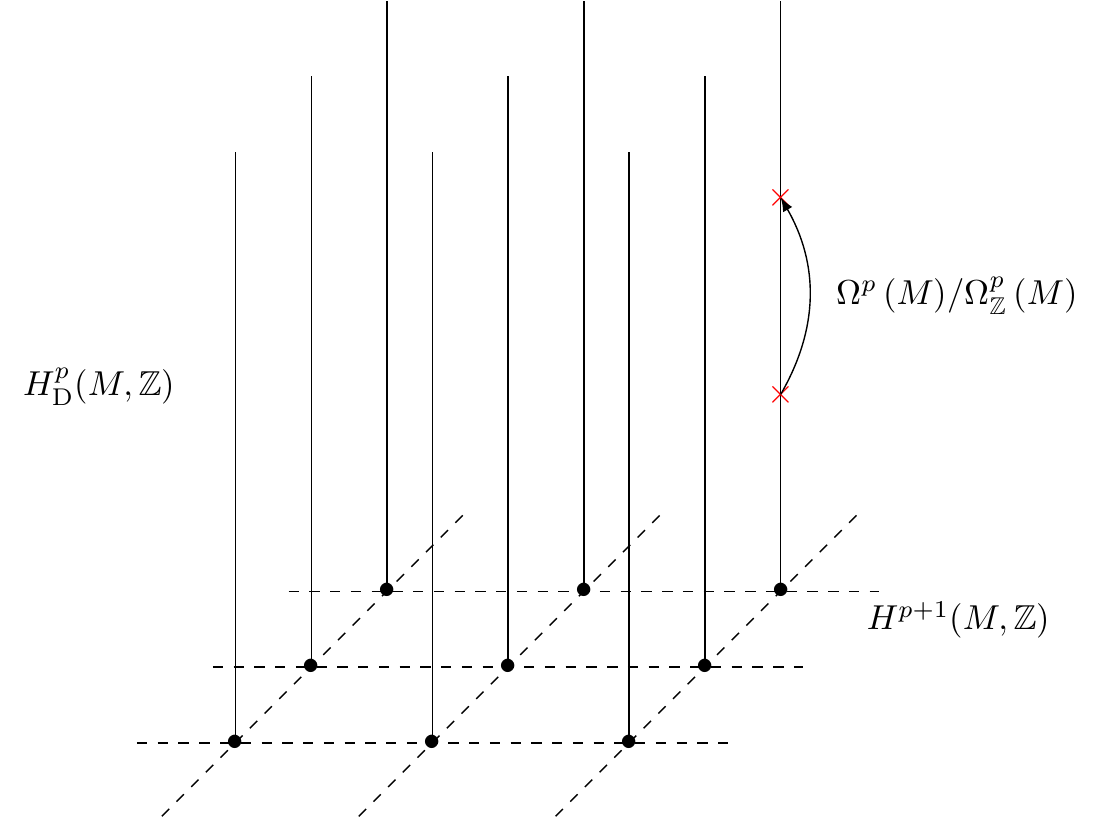}
			\caption{Exact sequence \eqref{short2}}
			%         \label{fig:y equals x}
		\end{subfigure}
		\\
		\vspace{1.cm}
		\begin{subfigure}[b]{1.0\textwidth}
			\centering
			\includegraphics[width=\textwidth]{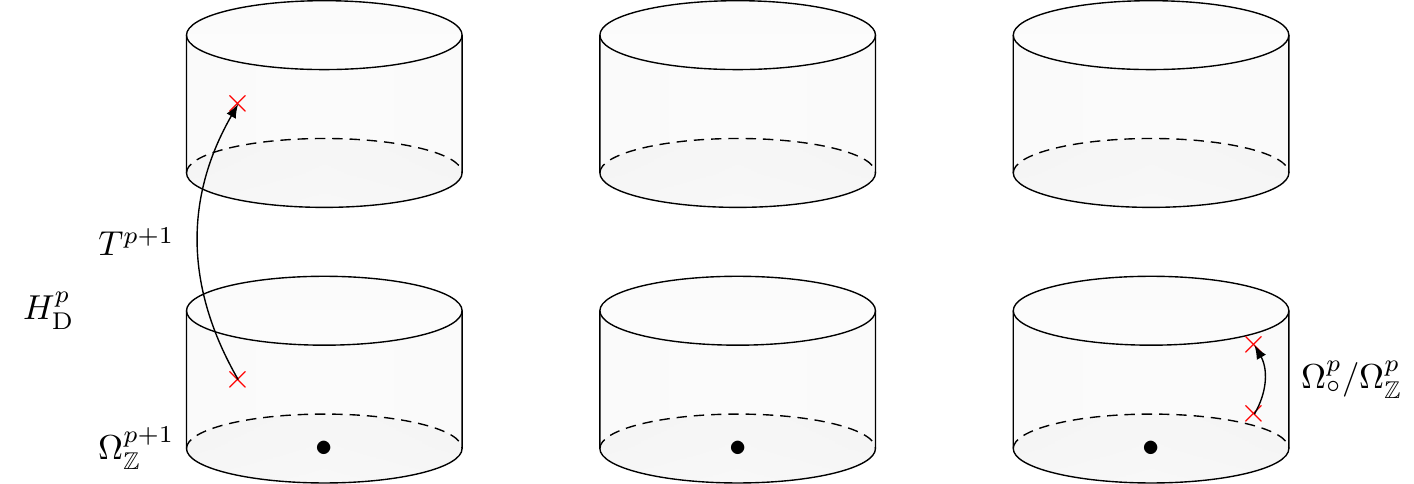}
			\caption{Exact sequence \eqref{short3}}
			%         \label{fig:three sin x}
		\end{subfigure}
		\caption{Representation of the exact sequences \eqref{short2} and \eqref{short3}}
		\label{Space_Fields}
	\end{figure}
	
	Let us compare \eqref{short2} and \eqref{short3}. On the one hand, since ${H^p}\left( M , \mathbb{R}/\mathbb{Z} \right) \simeq \left( \mathbb{R}/\mathbb{Z}\right)^{b_p} \oplus T_p$ and ${T_p} \simeq {T_{n - p - 1}} \simeq {T^{p + 1}}$ the torsion sector $T^{p + 1}$ is a part of the base space of $H_D^p(M)$ as described by \eqref{short2} and a part of the quantum field translation group of \eqref{short3}. On the other hand, the base space $\Omega^{p + 1}_\mathbb{Z}\left( M \right)$ has various connected components, each of which is parameterized by an element of $F^{p+1}$. Once an element of a connected component has been identified, any other element is obtained by adding to this first element an exact form. This implies that each connected component of $\Omega^{p + 1}_\mathbb{Z}\left( M \right)$ is isomorphic to $\Omega^p\left( M \right)/\Omega_\circ^p\left( M \right)$. So $\Omega^p\left( M \right)/\Omega_\mathbb{Z}^p\left( M \right)$ seems to be missing in \eqref{short3}, while it is $\left( \mathbb{R}/\mathbb{Z}\right)^{b_p}$ which seems to be missing in \eqref{short2}. This puzzling situation is solved by considering the following exact sequence
	\begin{equation}
		0 \to \Omega_\mathbb{Z}^p(M) \to \Omega_\circ^p(M) \to \frac{\Omega^p(M)}{\Omega_\mathbb{Z}^p(M)} \to \frac{\Omega^p(M)}{\Omega_\circ^p(M)} \to 0 \, ,
	\end{equation}
	which is a straightforward consequence of inclusions \eqref{inclusions1}. This sequence then reduces to
	\begin{equation}
		\label{short4}
		0 \to \frac{\Omega_\circ^p(M)}{\Omega_\mathbb{Z}^p(M)} \to \frac{\Omega^p(M)}{\Omega_\mathbb{Z}^p(M)} \to \frac{\Omega^p(M)}{\Omega_\circ^p(M)} \to 0 \, .
	\end{equation}
	Now, we see that in a more or less hidden way the same spaces are present in both exact sequences \eqref{short2} and \eqref{short3}. Let us point out that $\Omega_\circ^p\left( M \right)/\Omega_\mathbb{Z}^p\left( M \right) \simeq (\mathbb{R}/\mathbb{Z})^{b_p}$ which is a consequence of the isomorphisms mentioned after inclusions \eqref{inclusions1}. 
	
	Another remarkable point is that, since $H^0(M) \simeq \mathbb{Z}\simeq \Omega_\mathbb{Z}^0(M)$, exact sequences \eqref{short2} and \eqref{short3} extend to the case where $p = -1$. Thus, for consistency reason the following convention applies.
	
	\begin{definition}
		\begin{equation}
			H_D^{-1}(M) = \mathbb{Z} \, . 
		\end{equation}
	\end{definition}
	
	\noindent A gauge $0$-field is just a \v{C}ech $0$-cocycle, the cohomology class of which is the quantum field represented by this gauge field.
	
	\begin{exercise}
		Show that sequences \eqref{short2} and \eqref{short3} are exact and split.
	\end{exercise}
	
	\subsection{Integration of quantum fields}
	
	As they play a key role in defining the observables of the generalized $\mathrm{U}\!\left(1\right)$ BF theory, let us see how we can define integration of quantum fields along cycles of $M$. In the $3$-dimensional case, such an integral is deeply related to the so-called (Ehrenberg-Siday-)Aharonov-Bohm effect, from which stems charge quantization. Let us start by recalling that integration of a classical $p$-field, i.e., a $p$-form, along a $p$-cycles is well-defined \cite{BT82} and yields the pairing
	\begin{equation}
		\begin{aligned}
			\oint : \Omega ^p\left( M \right) \times {Z_p}\left( M \right) &  \to \mathbb{R} \\ 
			\left( \chi ,z \right) &  \mapsto \oint\limits_z \chi  \, . \\ 
		\end{aligned} 
	\end{equation}
	Furthermore, since a classical $p$-field $\chi^{(p)}$ gives rise to a gauge $p$-field $\chi^{[p]}$ defined by \eqref{classical}, it seems natural to set
	\begin{equation}
		\label{specialintcase}
		\oint_z \chi^{[p]} = \oint_z \chi^{(p)}\, .
	\end{equation}
	However, if $\rho^{(p)} \in \Omega_\mathbb{Z}^p(M)$ then on the one hand
	\begin{equation}
		\oint_z  \rho^{[p]} \in \mathbb{Z} \, ,
	\end{equation}
	and on the other hand $\rho^{[p]}$ is a gauge translation, i.e., $\bar{\rho}^{(p)} = 0$, so that $\boldsymbol{\rho}^{[p]} = \bar{\delta}(\bar{\rho}^{(p)}) = \boldsymbol{0} \in H_D^p(M)$. Consequently, if we want integration along $z$ to be well-defined on quantum fields, it  must be defined modulo integers, i.e., as an $\mathbb{R}/\mathbb{Z}$-valued linear functional on $H_D^p(M)$.
	
	Unfortunately, a gauge $p$-field does not necessarily derive from a classical field. This makes \eqref{specialintcase} improper for a generic gauge $p$-field. In fact, it is natural to assume that all the components of a gauge $p$-field may appear in the expression of its integral along a $p$-cycle, and because the components of a gauge $p$-field involve forms of various degrees, it can be expected that the $p$-cycle itself will have to give rise to a collection of local $p$-chains on which the components of the gauge $p$-field will be integrated. It turns out that such a decomposition of a $p$-cycle was originally introduced in order to prove the equivalence of singular and \v{C}ech homology \cite{We52,BT82}. Before explaining this, let us recall that we deal with de Rham $p$-chains, even locally. 
	
	\begin{definition} \label{suborcycledecompo}
		Let $z$ be a $p$-cycle of $M$. We set
		\begin{equation}
			U_{\alpha_0 \hdots \alpha_k}^{z} = \left\{\begin{gathered}
				U_{\alpha_0 \hdots \alpha_k} \hspace{0.25cm} \mbox{ when } \; U_{\alpha_0 \hdots \alpha_k} \cap z \neq  \varnothing\, , \hfill \\
				\varnothing  \hspace{1.2cm} \mbox{ otherwise.} \hfill \\ 
			\end{gathered}  \right.
		\end{equation}
		A \textbf{decomposition of $z$ subordinate to $\mathcal{U}_M$} is a $(p+2)$-tuple $z_{\{p\}} = \left(c_{(0,p)} , \hdots\right.$ $ \left.\hdots,c_{(p,0)} , c_{(p,-1)} \right)$ for which each $c^{\alpha_0 \hdots \alpha_k}_{(k,p-k)}$ is a $(p-k)$-chain has its support strictly contained in $U_{\alpha_0 \hdots \alpha_k}^z$ and such that
		\begin{equation}
			\label{decompcycle}
			\left\{ 
			\begin{gathered}
				\partial c_{(0,p)} = z \, , \hfill \\
				\partial c_{(k,p-k)} = b c_{(k-1,p-k+1)} \, , \hfill \\ 
				c_{(p,-1)} = b_0 c_{(p,0)} \, , \hfill \\
			\end{gathered} \right.
		\end{equation}
		for $k \in \{1, \hdots , p\}$.
	\end{definition}
	
	In descent \eqref{decompcycle}, $b$ and $\partial$, respectively denote the boundary operator on chains ($b^2=0$) and the \v{C}ech boundary operator ($\partial^2 = 0$), with $\partial$ defined on a collection $c_{(k,p-k)}$ according to
	\begin{equation}
		\label{defCechpartial}
		\left( \partial c_{(k,p-k)} \right)^{\alpha_0 \hdots \alpha_{k-1}} = \sum\limits_\alpha  c_{(k,p-k)}^{\alpha \alpha_0 \hdots \alpha_{k-1}} \, .
	\end{equation}
	Moreover, $\partial$ fulfills
	\begin{equation}
		\label{dualCechoperators}
		\delta a^{(p-1,-1)} [c_{(p,-1)}] = a^{(p-1,-1)} [\partial c_{(p,-1)}] \, ,
	\end{equation}
	for any \v{C}ech $p$-chain $c_{(p,-1)}$ and any \v{C}ech $(p-1)$-cochain $a^{(p-1,-1)}$, so that $\delta$ is the dual of $\partial$. 
	
	Intuitively, the first line in \eqref{decompcycle} describes a division of $z$ into ``pieces", the collection of which is denoted by $c_{(0,p)}$. The boundaries of the elements of this collection are themselves cut into pieces which yields a new collection $c_{(1,p-1)}$. This procedure is repeated until a collection $c_{(p,0)}$ of local $0$-chain is obtained. Then, each $0$-chain $c^{\alpha_0 \hdots \alpha_p}_{(p,0)}$ is decomposed along some chosen points in $U_{\alpha_0 \hdots \alpha_p}$. By construction, the coefficients of this decomposition are integers, and their sum is called the \textbf{degree} of $c_{(p,0)}$. The operator $b_0:C_0(M) \to \mathbb{Z}$ simply associates to a $0$-chain its degree. It is easy to check that the degree of a boundary is zero so that we have $b_0 b = 0$. Now, the last component of $z_{\{p\}}$ is the collection of integers $c_{(p,-1)} = b_0 c_{(p,0)}$, for which the sequence of identities $\partial c_{(p,-1)} = \partial b_0 c_{(p,0)} = b_0 \partial c_{(p,0)} = b_0 b c_{(p-1,1)} = 0$ implies that it is a \v{C}ech $p$-cycle of $\mathcal{U}_M$. A given $p$-cycle $z$ of $M$ admits an infinite number of different decompositions as above, any two of them differing by a $(p+2)$-tuple of the form
	\begin{equation}
		\left( \partial h_{(1,p)} ,  b h_{(1,p)} + \partial h_{(2,p-1)} , \hdots , b h_{(p,1)} + \partial h_{(p+1,0)} , b_0 \partial h_{(p+1,0)} \right) \, ,
	\end{equation}
	with $b_0 \partial h_{(p+1,0)} = \partial b_0 h_{(p+1,0)}$ a \v{C}ech boundary \cite{We52}.
	
	When seeing the perfect matching in all degrees of the components of a decomposition of a $p$-cycle and the components of a gauge $p$-field, the following definition seems natural.
	
	\begin{definition} \label{gaugefieldsintegration}
		Let $A^{[p]} = \left( A^{(0,p)} , A^{(1,p-1)} , \hdots , A^{(p,0)} , a^{(p+1,-1)} \right)$ be a gauge $p$-field and let $z_{\{p\}} = \left(c_{(0,p)} , \hdots , c_{(p,0)} , c_{(p,-1)} \right)$ be a decomposition of a $p$-cycle $z$. The \textbf{quantum integral} of the quantum field $\boldsymbol{A}^{[p]}$ along $z$ is defined as
		\begin{equation}
			\label{defintegral}
			\oint_z \boldsymbol{A}^{[p]} \stackrel{\mathbb{R}/\mathbb{Z}}{=} \sum\limits_{k = 0}^p (-1)^k \int\limits_{c_{(k,p-k)}} A^{(k,p-k)} \, ,
		\end{equation}
		with
		\begin{equation}
			\label{subintegrals}
			\int\limits_{c_{\left(k, p-k\right)}} A^{\left( k,p-k \right)} = \frac{1}{(k+1)!} \sum\limits_{\alpha_0, \hdots, \alpha_k} \, \int\limits_{c_{\left( k,p-k \right)}^{\alpha_0 \hdots \alpha _k}} A_{\alpha _0 \hdots \alpha _k}^{\left(k,p - k\right)}\, ,
		\end{equation}
		and where $\stackrel{\mathbb{R}/\mathbb{Z}}{=}$ means equality in $\mathbb{R}/\mathbb{Z}$, i.e., equality modulo integers. By construction, the integral of a quantum $p$-field along a $p$-cycle is $\mathbb{R}/\mathbb{Z}$-valued. The right-hand side of \eqref{defintegral} will be referred to as the integral of the gauge field $A^{[p]}$ along the decomposition $z_{\{p\}}$ of the $p$-cycle $z$ of $M$. It can also be referred to as the evaluation of $z_{\{p\}}$ on $A^{[p]}$. Relation \eqref{defintegral} then defining the \textbf{quantum evaluation} of $z$ on $\boldsymbol{A}^{[p]}$.
	\end{definition}

	To check the consistency of the above definition, it must be shown that the quantum integral of $\boldsymbol{A}^{[p]}$ along $z$ depends neither on the chosen decomposition of $z$ nor on the chosen gauge $p$-field which represents $\boldsymbol{A}^{[p]}$. Furthermore, when $\boldsymbol{A}^{[p]}$ is the quantum field defined by a classical $p$-field $\chi^{(p)}$, expression \eqref{defintegral} reduces to \eqref{specialintcase}.
	
	\begin{exercise}
		Show that: 
		\begin{enumerate}[1)]
			\item the evaluation of the decomposition of a boundary on any gauge field is an integer;
			\item the evaluation of any decomposition of a cycle on a gauge field transformation is an integer. 
		\end{enumerate}
	\end{exercise}
	
	As it is  $\mathbb{R}/\mathbb{Z}$-valued, the quantum integral of a quantum $p$-field along a $p$-cycle defines an element of $\U1$ according to
	\begin{equation}
		h(\boldsymbol{A}^{[p]},z) = e^{2 i \pi \oint_z \boldsymbol{A}^{[p]}}\, .
	\end{equation}
	This quantity is usually called a $\U1$ holonomy or  Wilson loop in the language of quantum field theory. In this article, we will refer to it as a \textbf{quantum observable} to stress out the fact that it is derived from an $\mathbb{R}/\mathbb{Z}$-valued integral. A nice exercise consists of recovering the definition of a gauge $p$-field from its holonomy.
	
	\subsection{DB product}
	
	Integration defined as above yields an $\mathbb{R}/\mathbb{Z}$-valued pairing between $H_D^p(M)$ and $Z_p(M)$. However, there is another pairing in which $H_D^p(M)$ is involved: The DB product. Although it can be defined for any two sets of quantum fields, we only need here the particular DB product
	\begin{equation}
		\label{DBpairingfields}
		\star : H_D^{n-p-1}(M) \times H_D^{p}(M) \to H_D^{n}(M)\, , 
	\end{equation}
	knowing that, from exact sequences \eqref{short2} and \eqref{short3}, $H_D^{n}(M) \simeq \Omega^n/\Omega^n_{\mathbb{Z}} \simeq \mathbb{R}/\mathbb{Z}$.
	
	For reasons that will become increasingly clear, let us introduce the following terminology
	
	\begin{definition}
		The \textbf{quantum complement} of $p \in \left\{0 , \hdots , n-1 \right\}$ in $M$ is the integer
		\begin{equation}
			\label{quantumcomplementary}
			q = n-p-1 \in \left\{0 , \hdots , n-1 \right\} \, ,
		\end{equation}
		so that conversely $p$ is the quantum complement of $q$ in $M$. 
	\end{definition}
	
	\noindent Let us remark that the notion of quantum complement is already natural from the point of view of relation \eqref{ddaggervsd}, and hence is not specific to the DB context. The pairing \eqref{DBpairingfields} is then defined as follows.
	
	\begin{definition} \label{gaugefieldsDBproduct}
		Let $A^{[p]}$ and $B^{[q]}$ be two gauge fields, $p$ and $q$ being quantum complement to each other. The \textbf{DB product} of $A^{[p]}$ and $B^{[q]}$, $B^{[q]} \star A^{[p]}$, is the gauge $n$-field whose components are
		\begin{align}
			\label{starproduct}
			\nonumber
			\left( B^{[q]} \star A^{[p]} \right)&_{\alpha_0 \hdots \alpha_k}^{(k,n-k)}\\
			=& \left\{ \begin{gathered}
				B_{{\alpha _0} \hdots \;{\alpha _k}}^{(k,q-k)} \wedge dA_{{\alpha _k}}^{(0,p)} \hspace{1cm}\mbox{ for } 0 \leq k \leq q \, , \hfill \\
				b_{\alpha_0 \hdots \alpha_{q+1}}^{(q+1,-1)} A_{\alpha_{q+1} \hdots \alpha _k}^{(k-n+p,n-k)} \hspace{0.35cm} \mbox{ for } n - p = q+1 \leq k \leq n\, , \hfill \\
				b_{{\alpha _0} \hdots \;{\alpha _{n - p}}}^{\left( q+1, -1 \right)} a_{\alpha_{n-p} \hdots \alpha_{n+1}}^{(p+1,-1)} \hspace{0.52cm} \mbox{ for } k = n + 1 \, , \hfill \\ 
			\end{gathered}  \right.
		\end{align}
		where repeated indices are not summed over. 
	\end{definition}
	
	This product is $\epsilon_p$-symmetric with $\epsilon_p = (-1)^{(p+1)(n-p)} = \pm 1$ and extends to quantum fields. Furthermore, this definition of the DB product also holds true even if $p$ and $q$ are not quantum complement with each other. In this case, $B^{[q]} \star A^{[p]}$ is a gauge $(p+q+1)$-field \cite{CS73,HLZ,BGST05}.
	
	The collection generated by the last line of the right-hand side the above relation is nothing but the cup product $b^{(q+1,-1)} \smallsmile a^{(p+1,-1)}$ of the \v{C}ech cocycles $b^{(q+1,-1)}$ and $a^{(p+1,-1)}$. The result of this cup product is then a \v{C}ech $(n+1)$-cocycle which is actually trivial, i.e., a $\delta$-coboundary, for dimensional reasons. The second line in the above definition is a straightforward extension of the cup product, whereas the first line is a combination of the cup and wedge products.
	
	\begin{exercise}
		Show that:
		\begin{enumerate}[1)]
			\item the components of $B^{[q]} \star A^{[p]}$ fulfill descent equations \eqref{descent};
			\item if we perform a gauge field transformation on $A^{[p]}$ as well as one on $B^{[q]}$ then the result is a gauge field transformation of $B^{[q]} \star A^{[p]}$.
		\end{enumerate}    
	\end{exercise}	
	
	From the above exercise, we deduce the following property.
	
	\begin{property}
		The DB product of gauge fields extends to quantum fields, the DB product of $\boldsymbol{B}^{[q]}$ with $\boldsymbol{A}^{[p]}$ thus being a quantum $n$-field denoted by $\boldsymbol{B}^{[q]} \star \boldsymbol{A}^{[p]}$.
	\end{property}

	We recognize in the first line of the right-hand side of \eqref{starproduct} a collection of local classical BF Lagrangians, and since $\boldsymbol{B}^{[q]} \star \boldsymbol{A}^{[p]}$ is a quantum field, from a physical point of view it is legitimate to call this product of quantum fields a \textbf{quantum Lagrangian}. As we have identified a Lagrangian, the next step is to define an action. Since, by definition, a quantum Lagrangian is a quantum $n$-field, it is natural to define the \textbf{quantum action} of the generalized $\U1$ BF theory as the quantum integral over $M$ of the quantum Lagrangian
	\begin{equation}
		S_{BF} \left( \boldsymbol{A}^{[p]},\boldsymbol{B}^{[q]} \right) = \oint_M \boldsymbol{B}^{[q]} \star \boldsymbol{A}^{[p]} \, .
	\end{equation}
	So, as a quantum integral the action $S_{BF}$ is $\mathbb{R}/\mathbb{Z}$-valued. This is another justification of our extensive and systematic use of the word ``quantum'' all along the construction, whether we are talking about the fields, the Lagrangian and now the action. If $M_{\{n\}} = \left( M_{(0,n)} , \hdots , M_{(n,0)} , m_{(n,-1)}  \right)$ is a polyhedral decomposition of $M$ subordinate to $\mathcal{U}_M$, and if $A^{[p]}$ and $B^{[q]}$ are gauge fields representing $\boldsymbol{A}^{[p]}$ and $\boldsymbol{B}^{[q]}$, respectively, then we have
	\begin{equation}
		\label{defintaction}
		\oint_M \boldsymbol{B}^{[q]} \star \boldsymbol{A}^{[p]} \stackrel{\mathbb{R}/\mathbb{Z}}{=} \sum\limits_{k = 0}^n (-1)^k \int\limits_{M_{({k,n-k})}} \left( B^{[q]} \star A^{[p]} \right)^{(k,n-k)} \, .
	\end{equation}

	Let us remark that the quantum action $S_{BF}$ defines a generalized $\U1$ holonomy according to
	\begin{equation}
		h(\boldsymbol{A}^{[p]} \star \boldsymbol{B}^{[q]},M) = e^{2 i \pi S_{BF}\left( \boldsymbol{A}^{[p]},\boldsymbol{B}^{[q]} \right)}\, .
	\end{equation}
	This generalized holonomy will play the role of a formal measure density in the functional integration framework of the generalized $\U1$ BF theory. Here again, it is possible to recover the expression of the DB product of two gauge fields from this generalized holonomy, an exercise left to the reader. We can introduce a coupling constant $k$, and define the action of the corresponding theory
	\begin{equation}
		S_{BF,k} \left( \boldsymbol{A}^{[p]},\boldsymbol{B}^{[q]} \right) = k S_{BF} \left( \boldsymbol{A}^{[p]},\boldsymbol{B}^{[q]} \right)\, .
	\end{equation}
	The action $S_{BF}$ being $\mathbb{R}/\mathbb{Z}$-valued, the coupling constant $k$ must be an integer for $S_{BF,k}$ to be well-defined. This is usually referred to as the quantization of the coupling constant. Let us point out that, if we use classical fields in the construction of the generalized BF action, then we will never obtain a quantization of the coupling constant. In this case, the action is classical and reads $\oint_M B \wedge d A$. It is also quite obvious that the ``large gauge group" of this classical action is $\Omega_\circ^p(M) \times \Omega_\circ^q(M)$ and the classical action is strictly gauge invariant, unlike the quantum action which is defined modulo integers. This is why the coupling constant is not quantized in the classical context.
	
	\section{Quantum currents}
	\label{section_quantum_currents}
	
	There are several reasons which justify the extension of the previous construction to ``singular" objects. Some of them were already mentioned in the introduction. In addition, let us recall that in the axiomatic approach of quantum field theory, as the one of Wightman and Garding for instance \cite{WG65}, fields turn out to be operator-valued distributions. Even in the Feynman path integral approach to quantum field theory, the use of singular fields is required. Indeed, it is well-known that  the Lebesgue functional measure is zero on the infinite dimensional classical configuration space. However, by adding to the classical configuration space some singular fields it is possible to construct an abstract Wiener space, i.e., an enlarged configuration space which can be endowed with a nontrivial Gaussian measure \cite{CM44,GS16,G67}. In fact, abstract Wiener spaces were originally introduced to give a rigorous mathematical background to the functional integral used in the study of the Brownian motion. The Feynman path integral can be seen as a formal extension of this construction first in quantum mechanics and then in quantum field theory.
	
	For later convenience, let us introduce the graded \v{C}ech operator $\hat{\delta}$ which acts on local de Rham $l$-currents according to
	\begin{equation}
		\hat{\delta} = (-1)^{n-l} \delta \, .
	\end{equation} 
	so that $\hat{\delta}$ is the dual of the operator
	\begin{equation}
		\hat{\partial} = (-1)^{n-l} \partial \, .
	\end{equation}
	The action of $\partial$ on a collection $C{_k}{^a}$ of local $a$-forms is derived from \eqref{defCechpartial} by considering a $a$-form as a $(n-a)$-current. Thus, we have 
	\begin{equation}
		\hat{\partial} C{_k}{^a} = (-1)^a \partial C{_k}{^a} \, .
	\end{equation}
	Let us point out that since
	\begin{equation}
		\hat{\delta} d^\dagger + d^\dagger \hat{\delta} = \delta d - d \delta = 0 \, ,
	\end{equation} 
	$d^\dagger + \hat{\delta}$ is a cohomology operator whereas $d + \delta$ is not, the latter being not nilpotent. Similarly, $d + \hat{\partial}$ is nilpotent whereas $d + \partial$ is not.
	
	The injection of real numbers into the set of de Rham $n$-currents, denoted by $d_{-1}^\dagger$, is defined by
	\begin{equation}
		\left(d_{-1}^\dagger h\right) [\omega] = h \oint_M \omega \, ,
	\end{equation}
	for any $h \in \mathbb{R}$ and any $\omega \in \Omega^n(M)$. Note that $d_{-1}^\dagger = d_{-1}$ which simply means that the constant $0$-form $d_{-1} h$ defines the constant $n$-current $d_{-1}^\dagger h$. Finally, we define the operator $i_n: \Omega^n(M) \rightarrow \mathbb{R}$ according to
	\begin{equation}
		i_n \omega = \oint_M \omega \, ,
	\end{equation}
	for any $\omega \in \Omega^n(M)$. This operator is clearly the dual of $d_{-1}^\dagger$ since
	\begin{equation}
		\left(d_{-1}^\dagger h\right) [\omega] = h [i_n \omega] \, .
	\end{equation}
	Furthermore, $d_{-1}^\dagger h \in \Omega_n^\mathbb{Z}(M)$ when $h \in \mathbb{Z}$. Eventually, if $c^{(p,-1)}$ is a real \v{C}ech $p$-cochain of $\mathcal{U}_M$ then $d_{-1}^\dagger c^{(p,-1)}$ will denote the collection of local de Rham $n$-currents whose component in $U_{\alpha_0 \dots \alpha_p}$ is the de Rham $n$-current $d_{-1}^\dagger c^{(p,-1)}_{\alpha_0 \dots \alpha_p}$, a current that must be evaluated on $n$-forms with compact support in $U_{\alpha_0 \dots \alpha_p}$. Similarly, if $\omega{_p}{^n}$ is a collection of $n$-forms such that each $\omega{_p}{^{n,\alpha_0 \dots \alpha_p}}$ has compact support in $U_{\alpha_0 \dots \alpha_p}$ then $i_n \omega^{(p,n)}$ is the \underline{real} \v{C}ech $p$-cochain of $\mathcal{U}_M$ whose component in $U_{\alpha_0 \dots \alpha_p}$ is the real number $i_n \omega{_p}{^{n, \alpha_0 \dots \alpha_p}}$. Moreover, $i_n \omega{_p}{^n}$ defines a \v{C}ech $p$-cochain of $\mathcal{U}_M$if and only if each $\omega{_p}{^{n, \alpha_0 \dots \alpha_p}}$ belongs to $\Omega_\mathbb{Z}^n(U_{\alpha_0 \dots \alpha_p})$.

	\subsection{Quantum currents and associated exact sequences}
	
	By simply replacing local forms with local de Rham currents in the previous construction of gauge fields, we obtain the following definition.
	
	\begin{definition} \label{gaugecurrent}
		Let $p$ and $q$ be quantum complement integers and let $A_{[q]} = \left(  A{^0}{_{q+1}} , A{^1}{_{q+2}} , \hdots , A{^{p}}{_n} , a^{(p+1,-1)} \right)$ be a $(p+2)$-tuple in which each $A{^k}{_{q+1+k}}$ is a collection of local de Rham $(q+k)$-currents, one in each non-empty $U_{\alpha_0 \hdots \alpha_k}$, and $a^{(p+1,-1)}$ is a \v{C}ech $(p+1)$-cocycle of $\mathcal{U}_M$. We say that $A_{[q]}$ is a \textbf{gauge $q$-current} of $\mathcal{U}_M$ if its components fulfill the following decent equations
		\begin{equation}
			\label{descentcurrent}
			\left\{ \begin{aligned}
				& \hat{\delta} A{^{k-1}}{_{q+k}} = d^\dagger A{^k}{_{q+k+1}} \quad \mbox{for $k \in \{0, \hdots , p \}$}, \\
				& \hat{\delta} A{^p}{_n} = d^\dagger_{-1} a^{(n-q,-1)} \, . 
			\end{aligned} \right.
		\end{equation}
	\end{definition}
	
	The above descent equations must be considered in the appropriate intersections of $\mathcal{U}_M$, each $A{^k}{_{q+k+1,\alpha_0 \hdots \alpha_k}}$ being an element of the topological dual of $\Omega^{q+k+1}_c(U_{\alpha_0 \hdots \alpha_k})$, the space of $(q+k+1)$-forms with compact support in $U_{\alpha_0 \hdots \alpha_k}$, and not of $\Omega^{q+k+1}(U_{\alpha_0 \hdots \alpha_k})$. The use of $\hat{\delta}$ in \eqref{descentcurrent} thus ensures that the descent equations reduce to descent equations \eqref{descent} when local de Rham currents are replaced by local forms (see below). The fact that the first component of a gauge $q$-current is a collection of local $(q+1)$-currents may seem strange. The reason for this discrepancy will be explained in a moment. To get ride of this discrepancy, it is possible to replace $\hat{\delta}$ and $d^\dagger$ by their expression in terms of $\delta$ and $d$, respectively, in descent equations \eqref{descentcurrent}, these equations then taking the form \eqref{descent}. We didn't make this choice of writing in order to keep the duality between currents and forms more obvious and also because the dimension of a de Rham current is more suitable for our purpose than its degree.
	
	The next step is to define classes of gauge currents following what we did for gauge fields. 
	
	\begin{definition} \label{gaugecurrenttransfo}
		A \textbf{gauge $q$-current transformation} is a gauge $q$-current of the form
		\begin{equation}
			\label{gaugetransfocurrent}
			\left( d^\dagger G{^0}{_{q+2}} , \hat{\delta} G{^0}{_{q+2}} + d^\dagger G{^1}{_{q+3}} , \hdots , \hat{\delta} G{^{p-1}}_{n} + d_{-1}^\dagger g^{(p,-1)} , \delta g^{(p,-1)} \right) \, ,
		\end{equation}
		with $g^{(p,-1)}$ a \v{C}ech $p$-cochain of $\mathcal{U}_M$, $p$ being the quantum complement of $q$. Two gauge $q$-currents which differ by a gauge $q$-current transformation are said to be equivalent. The equivalence class of a gauge $q$-current $A_{[q]}$ is called a \textbf{quantum $q$-current} and denoted by $\boldsymbol{A}_{[q]}$. The set of quantum $q$-currents of $M$, $H_q^D(M)$, is a $\mathbb{Z}$-module. 
	\end{definition}
	
	Like for quantum fields, there are exact sequences into which $H_q^D(M)$ sits. 
	
	\begin{theorem}
		The space $H^D_q(M)$ sits into the short exact sequence
		\begin{equation}
			\label{currentshort1}
			0 \to \frac{\Omega_{q+1}(M)}{\Omega_{q+1}^\mathbb{Z}(M)} \xrightarrow{\bar{\delta}} H^D_q(M) \xrightarrow{cl} H^{p+1}(M) \to 0 \, .
		\end{equation}
	\end{theorem}
	
	A proof of this theorem can be straightforwardly obtained from the proof of \eqref{short2} \cite{Br93,BGST05}. The morphisms $cl$ and $\bar{\delta}$ appearing in \eqref{currentshort1} are natural extensions of the morphisms $cl$ and $\delta$ which appear in \eqref{short2}. In particular, if $\chi_{(q)}$ is a de Rham $q$-current of $M$, we set
	\begin{equation}
		\label{classicalcurrent}
		\chi_{[q]} = \left( \delta_{-1} \chi_{(q)} , 0 , \hdots ,0 \right)\, ,
	\end{equation}
	with $\delta_{-1}$ still denoting the restriction to open subsets of $M$ which is defined by setting
	\begin{equation}
		( \delta_{-1} \chi_{(q)}) [\omega] = \chi_{(q)} [\omega] \, ,
	\end{equation}
	for all $\omega \in \Omega^q_c(U_\alpha)$. The above relation is meaningful because a $q$-form with compact support in $U_\alpha$ trivially defines a $q$-form of $M$. 
	
	\vspace{0.5cm}
	
	As already noticed, if $A_{[q]}$ is a gauge $q$-current then there exists a classical $q$-current $F_{(q)}$ such that $d^\dagger A_{q+1}^0 = \hat{\delta}_{-1} F_{(q)}$. This classical $q$-current is called the \textbf{curvature} of $A_{[q]}$. A gauge $q$-current whose curvature is zero is said to be \textbf{flat}. Then, since gauge current transformations are obviously flat, any equivalent gauge current has the same curvature and we can talk about the \textbf{curvature of a quantum $q$-current}. The curvature of a quantum current is necessarily $d^\dagger$-closed since it is locally $d^\dagger$-exact. Moreover, as already mentioned, if we add to a gauge $q$-current its curvature as new first component, then we obtain a \v{C}ech-de Rham descent of this curvature, and since this descent, by construction, ends with a(n integral) \v{C}ech cocycle, then the curvature is automatically $\mathbb{Z}$-valued on $\Omega_\mathbb{Z}^q(M)$. Thus, curvatures are elements of $\Omega^\mathbb{Z}_q(M)$ and we have a map $\bar{d}: H_q^D(M) \to \Omega^\mathbb{Z}_q(M)$. This map is surjective because any element $\Omega^\mathbb{Z}_q(M)$ gives rise, through a \v{C}ech-de Rham descent, to at least one gauge $q$-current. This means we have the exact sequence
	\begin{equation}
		H_q^D(M) \xrightarrow{\bar{\delta}} \Omega^\mathbb{Z}_q(M) \to 0 \, .
	\end{equation}
	This exact sequence can be enlarged to the left, thus yielding the following theorem, the proof of which is left as an exercise.
	\begin{theorem}
		The space $H^D_q(M)$ sits into the short exact sequence
		\begin{equation}
			\label{currentshort2}
			0 \to H^p(M,\mathbb{R}/\mathbb{Z}) \xrightarrow{i} H^D_q(M) \xrightarrow{\bar{d}} \Omega^\mathbb{Z}_q(M) \to 0 \, .
		\end{equation}
	\end{theorem}
	
	Here again, if we consider gauge fields instead of gauge currents, then $\Omega^\mathbb{Z}_q(M)$ becomes $\Omega^{n-q}_\mathbb{Z}(M) = \Omega^{p+1}_\mathbb{Z}(M)$, $p$ and $q$ being quantum complement to each other, and we obtain exact sequence \eqref{short3} instead of \eqref{currentshort2}. Note also the consistency of our terminology concerning curvatures since $\Omega^{p+1}_\mathbb{Z}(M)$ can be seen as a subset of $\Omega^\mathbb{Z}_q(M)$. Finally, we have not discussed the topology on the space of quantum fields and quantum currents yet. The exact sequences into which these spaces sit allow to provide each of them with an inherited topology \cite{HLZ}.
	
	\begin{exercise}
		Show that sequences \eqref{currentshort1} and \eqref{currentshort2} are exact.
	\end{exercise}
	
	Now, if we compare exact sequences \eqref{short2} and \eqref{currentshort1} on the one hand, and exact sequences \eqref{short3} and \eqref{currentshort2} on the other hand, we are not surprised by the following property.
	
	\begin{property} For each pair $(p,q)$ of quantum complement integers, there is a canonical injection
		\begin{equation}
			H_D^p(M) \hookrightarrow H^D_q(M) \, ,
		\end{equation}
		that we refer to as an inclusion.
	\end{property}
	
	\begin{proof}
		Let $A^{[p]} = \left( A^{\left( 0,p \right)} , A^{\left( 1,p-1 \right)} , \hdots , A^{\left( p,0 \right)} , a^{\left( p+1,-1 \right)} \right)$ be a gauge $p$-field and let $q$ be the quantum complement of $p$. For $k \in \{0 , \hdots , p \}$, we set
		\begin{equation}
			A{^k}{_{q+k+1,\alpha_0 \hdots \alpha_k}} [\omega] = \oint_M A_{\alpha_0 \hdots \alpha_k}^{(k,p-k)} \wedge \omega ,
		\end{equation}
		for any $\omega \in \Omega_c^{q+k+1}(U_{\alpha_0 \hdots \alpha_k})$. This yields the $(n-q+1)$-tuple  
		$A_{[q]} = \left( A{^0}{_{q+1}} , \hdots , A{^p}{_n} , a^{(n-q,-1)} \right)$. Then, by evaluating each $d^\dagger A{^k}{_{n-p+k,\alpha_0 \hdots \alpha_k}}$ on a test form, i.e., a $(n-p+k)$-form with compact support in $U_{\alpha_0 \hdots \alpha_k}$, and by using integration by parts, it is not difficult to check that the components of $A_{[q]}$ fulfill the descent equations \eqref{descentcurrent}, thus showing that $A_{[q]}$ is a gauge $q$-current. 
	\end{proof}
	
	\subsection{Dual quantum fields and evaluation of quantum currents}
	
	In the previous section, we saw that integration of a quantum $p$-field is an $\mathbb{R}/\mathbb{Z}$-valued functional on $Z_p(M)$, which can be obtained from integration of gauge $p$-fields along decompositions of cycles subordinate to $\mathcal{U}_M$. We now present the equivalent for quantum currents. We must define the objects on which quantum currents will be evaluated, simply requiring that this evaluation is an $\mathbb{R}/\mathbb{Z}$-valued linear functional defined from an evaluation of the corresponding gauge currents. Thus, let us proceed in complete analogy with integration of quantum fields.
	
	\begin{definition} \label{dualgaugefield}
		Let $p$ and $q$ be quantum complement integers and let $C^{\{q\}} = \left( C{_0}{^{q+1}} , \hdots , C{_p}{^n} ,  c_{(p,-1)} \right)$ be a $(p+2)$-tuple where each $C{_k}{^{q+1+k}}$ is a collection of $(q+1+k)$-forms $C{_k}{^{q+1+k ,\alpha_0 \hdots \alpha_k}}$, one for each non-empty $U_{\alpha_0 \hdots \alpha_k}$, which have compact support in $U_{\alpha_0 \hdots \alpha_k}$, $c_{(p,-1)}$ being a \v{C}ech chain of $\mathcal{U}_M$. We say that $C^{\{q\}}$ is a \textbf{dual gauge $q$-field} of $\mathcal{U}_M$ if its components fulfill the following descent equations
		\begin{equation}
			\label{descentdualfield}
			\left\{\begin{aligned}
				& d C{_k}{^{q+1+k}} = \hat{\partial} C{_{k+1}}{^{q+2+k}} \quad \mbox{for $k \in \{0, \hdots , p - 1 \}$}, \\
				& i_n C{_{p}}{^n} = c_{(p,-1)} \, .
			\end{aligned}\right.
		\end{equation}
	\end{definition}
	
	Let us stress out the main differences between dual gauge fields and gauge fields. Firstly, the forms defining a dual gauge field all have a local compact support unlike the forms defining a gauge field. Secondly, the descent equations of a dual gauge field use the homology operator $\hat{\partial}$ and not the cohomology operator $\delta$. The compact supports will ensure that the evaluations with the local de Rham currents are all well-defined whereas the use of $\hat{\partial}$ will ensure the duality between the descent equations of gauge currents and those of dual gauge fields. 
	
	The first and last components of a dual gauge field both have a specific property.
	
	\begin{property}
		\label{prop_dual_gauge_q_field}
		Let $C^{\{q\}} = \left( C{_0}{^{q+1}} , \hdots , C{_p}{^n} , c_{(p,-1)} \right)$ be a dual gauge $q$-field. Then
		\begin{equation}
			\label{dualfieldednings}
			\left\{ \begin{gathered}
				\partial C{_0}{^{q+1}} \in \Omega_\mathbb{Z}^{q+1}(M) \, , \hfill \\
				\partial c_{(p,-1)} = 0 \, . \hfill
			\end{gathered} \right.
		\end{equation}
	\end{property}
	
	\begin{proof}
		Firstly, $\partial C{_0}{^{q+1}}$ is a closed $(q+1)$-form since $d \partial C{_0}{^{q+1}} = \partial d C{_0}{^{q+1}} = \partial^2 C{_1}{^{q+2}} = 0$. Let us consider $\omega^{(p)} \in \Omega_\mathbb{Z}^{p}(M)$ together with a family $\omega^{(k,p-k)}$ of local forms generated by a \v{C}ech-de Rham descent of $\omega^{(p)}$ with respect to $\mathcal{U}_M$, this descent ending with the \v{C}ech cocycle $n^{(p,-1)}$. Then, we have
		\begin{align}
			\oint_M \omega^{(p)} \wedge \left(\partial C{_0}{^{q+1}}\right) 
			&= \oint_M \omega^{(p)} \wedge \left(\sum_\alpha C{_0}{^{q+1,\alpha}}\right)  \nonumber \\
			&= \sum_\alpha \oint_M \left(\delta_{-1} \omega^{(p)}\right)_\alpha \wedge C{_0}{^{q+1,\alpha}}  \nonumber \\
			&= \sum_\alpha \oint_M d \omega^{(0,p-1)}_\alpha \wedge C{_0}{^{q+1,\alpha}} \nonumber \\
			&= (-1)^p \sum_\alpha \oint_M \omega^{(0,p-1)}_\alpha \wedge d C{_0}{^{q+1,\alpha}}  \nonumber \\
			&= (-1)^p \sum_\beta \oint_M \omega^{(0,p-1)}_\beta \wedge \left(\sum_\alpha C{_1}{^{q+2,\alpha \beta}}\right)  \nonumber \\
			&= (-1)^p \frac{1}{2} \sum_{\alpha \beta} \oint_M \left(\delta \omega^{(0,p-1)}\right)_{\alpha \beta} \wedge C{_1}{^{q+2,\alpha \beta}}  \nonumber \\
			\oint_M \omega^{(p)} \wedge (\partial C{_0}{^{q+1}})            
			&= (-1)^p \frac{1}{2} \sum_{\alpha \beta} \oint_M d \omega^{(1,p-2)}_{\alpha \beta} \wedge C{_1}{^{q+2,\alpha \beta}} \, .
		\end{align}
		By continuing this procedure, we eventually obtain
		\begin{equation}
			\oint_M \omega^{(p)} \wedge (\partial C{_0}{^{q+1}})
			= \left(-1\right)^{\frac{p\left(p+1\right)}{2}} n^{(p,-1)} [c_{(p,-1)}] \in \mathbb{Z} \, ,
		\end{equation}
		for any $\omega^{(p)} \in \Omega_\mathbb{Z}^{p}(M)$, which implies that $\partial C{_0}{^{q+1}} \in \Omega_\mathbb{Z}^{q+1}(M)$.
		
		Secondly, $\partial c_{(p,-1)} = \partial i_n C{_p}{^n} = i_n d C{_{p-1}}{^{n-1}} = \oint_M d C{_{p-1}}{^{n-1}} = 0$ since, by construction, all the local forms defining $C{_{p-1}}{^{n-1}}$ have local compact support.
	\end{proof}
	
	\noindent By adding $\partial C{_0}{^{q+1}}$ to the collection of local forms defining the dual gauge field $C^{\{q\}}$, we obtain a \v{C}ech-Weil descent of $\partial C{_0}{^{q+1}}$ \cite{W51}. This kind of descents can be used to prove that the de Rham cohomology of closed $k$-forms with integral periods is isomorphic to $F_{n-k}$.
	
	As we did for gauge fields and gauge currents, we now define special dual gauge fields that will generate an equivalence relation for dual gauge fields.
	
	\begin{definition} \label{dualgaugetransfo}
		A \textbf{dual gauge $q$-field transformation} is a dual gauge $q$-field of the form
		\begin{equation}
			\left( \hat{\partial} G{_1}{^{q+1}} ,  d G{_1}{^{q+1}} + \hat{\partial} G{_2}{^{q+2}} , \hdots , d G{_p}{^{n-1}} + \hat{\partial} G{_{p+1}}{^n} , i_n \hat{\partial} G{_{p+1}}{^n} \right) \, ,
		\end{equation}
		where $G{_{p+1}}{^n}$ is such that $i_n G{_{p+1}}{^n}$ is a \v{C}ech $(p+1)$-chain of $\mathcal{U}_M$. Two dual gauge $q$-fields which differ by a dual gauge $q$-field transformation are said to be equivalent. The equivalence class of a dual gauge $p$-field $C^{\{q\}}$ is called a \textbf{dual quantum $q$-field} of $M$ and denoted by $\boldsymbol{C}^{\{q\}}$. The set of dual quantum $q$-fields of $M$ is a $\mathbb{Z}$-module denoted by $D_H^q(M)$.
	\end{definition}
	
	Let us note that the constraint that $i_n G{_{p+1}}{^n}$ is a \v{C}ech $(p+1)$-chain of $\mathcal{U}_M$ can be compared with the fact that the penultimate component of a gauge field transformation contains a \v{C}ech $(p+1)$-chain of $\mathcal{U}_M$.
	
	The $\mathbb{Z}$-module $D_H^q(M)$ sits into two exact sequences which look like very much the two exact sequences into which $H_D^p(M)$ is sitting.
	
	\begin{theorem}
		The $\mathbb{Z}$-module $D_H^q(M)$ sits in the following exact sequences
		\begin{equation}
			\label{dualfieldsequences}
			\begin{gathered}
				0 \to \frac{\Omega^q(M)}{\Omega_\mathbb{Z}^q(M)} \xrightarrow{\bar{\mu}} D_H^q(M) \xrightarrow{cl} H_{p}(M) \to 0 \, , \hfill \\
				0 \to H_{p+1}(M,\mathbb{R}/\mathbb{Z}) \xrightarrow{i} D_H^q(M) \xrightarrow{\bar{\partial}} \Omega^{q+1}_\mathbb{Z}(M) \to 0 \, . \hfill
			\end{gathered}
		\end{equation}
	\end{theorem}
	
	\begin{proof}
		We postpone the proof concerning the first sequence of this theorem to the next section where we will meet the same kind of construction of exact sequences for dual quantum currents. Nevertheless, for later convenience let us define here the morphism $cl$. By construction, the last component of a dual gauge $q$-field is a \v{C}ech $p$-cycle, the last component of a dual gauge field transformation being more specifically a \v{C}ech boundary. This means that we can associate to any dual quantum $q$-field an element of $H_{p}(M)$. This defines the morphism $cl$.
		
		Now, let us concentrate on the second sequence of \eqref{dualfieldsequences}. We have already noticed that the first component of a dual gauge $q$-field $C^{\{q\}}$ defines a closed $(q+1)$-form $\partial G{_{0}}{^{q+1}}$ with integral period. Moreover, the closed form with integral period defined by a dual gauge $q$-field transformation is obviously the zero form. Hence, we can associate to any dual quantum $q$-field a closed $(q+1)$-form with integral period which yields the morphism $\bar{\partial}$. The morphism $i$ is defined as follows. For any $\boldsymbol{\zeta} \in H_{p+1}(M,\mathbb{R}/\mathbb{Z})$, let $\zeta$ be a real \v{C}ech $(p+1)$-chain of $\mathcal{U}_M$ which represents $\boldsymbol{\zeta}$ and let $G{_{p+1}}{^n}$ be a collection of $n$-forms with local compact support such that $i_n G{_{p+1}}{^n} = \zeta$. The elements of $G{_{p+1}}{^n}$ are referred to as bump forms in \cite{BT82}. Then, $\partial \zeta$ is a \v{C}ech $p$-cycle of $\mathcal{U}_M$ and $\left( 0 , \hdots , 0 , \partial G{_{p+1}}{^n} , \partial \zeta \right)$ is a dual gauge $q$-field of $\mathcal{U}_M$. Obviously, other representatives of $\boldsymbol{\zeta}$ differ from $\zeta$ by a \v{C}ech $(p+1)$-chain $m$, the difference between the two dual gauge fields defined by these representatives thus being of the form $\left( 0 , \hdots , 0 , \partial H{_{p+1}}{^n} , \partial m \right)$ with $i_n H{_{p+1}}{^n} = m$, so that this last dual gauge field is actually a dual gauge transformation. Consequently, we can associate to $\boldsymbol{\zeta}$ the dual quantum field of $\left( 0 , \hdots , 0 , \partial G{_{p+1}}{^n} , \partial \zeta \right)$. This association defines the morphism $i$. Now, the image by $\bar{\partial}$ of $i(\boldsymbol{\zeta})$ is obviously zero, so that $\Ima i \subset \ker \hat{\partial}$. Conversely, if $\boldsymbol{C}^{\{q\}} \in \ker \bar{\partial}$, then it is flat, i.e., $\bar{\partial} \boldsymbol{C}^{\{q\}} = \partial C{_0}{^{q+1}} = 0$, which implies that $C{_0}{^{q+1}} = \hat{\partial} G{_1}{^{q+1}}$. Now, by subtracting to $C^{\{q\}}$ the dual gauge field $\left( \hat{\partial} G{_1}{^{q+1}} , d G{_1}{^{q+1}} , 0 , \dots , 0\right)$, we obtain a new representative of $\boldsymbol{C}^{\{q\}}$, whose first component is zero and whose second component satisfies $\partial C{_1}{^{q+2}} = 0$, which implies that $C{_1}{^{q+2}} = \hat{\partial} G{_2}{^{q+2}}$. By repeating this procedure, we end with a representative $C^{\{q\}} = \left( 0 , \dots , 0 , C{_p}{^n} ,  c_{(p,-1)} \right)$ of $\boldsymbol{C}^{\{q\}}$ which satisfies $\partial C{_p}{^n} = 0$, and hence $C{_p}{^n} = \hat{\partial} C{_{p+1}}{^n}$. If $i_n C{_{p+1}}{^n}$ is a \v{C}ech chain then $C^{\{q\}}$ is a dual gauge field transformation. Therefore, as $\partial (i_n C{_{p+1}}{^n})$ is by construction a \v{C}ech cycle, $\boldsymbol{C}^{\{q\}}$ is not trivial if and only if $i_n C{_{p+1}}{^n}$ is a representative of an element of $H_{p+1}(M,\mathbb{R}/\mathbb{Z})$, which shows that $\ker \bar{\partial} \subset \Ima i$, and thus that the second exact sequence of \eqref{dualfieldsequences} is exact.
	\end{proof}
	
	We can now define the evaluation of quantum currents on dual quantum fields, starting with the evaluation of gauge currents on dual gauge fields.
	
	\begin{definition}
		Let $A_{[q]} = \left(  A{^0}{_{q+1}} , A{^1}{_{q+2}} , \hdots , A{^{p}}{_n} , a^{(p+1,-1)} \right)$ be a gauge $q$-current and let $C^{\{q\}} = \left( C{_0}{^{q+1}} , \hdots , C{_p}{^n} ,  c_{(p,-1)} \right)$ be a dual gauge $q$-field. The evaluation of $A_{[q]}$ on $C^{\{q\}}$ is defined by
		\begin{equation}
			A_{[q]} [ C^{\{q\}} ] = \sum_{k = 0}^p (-1)^k A{^k}{_{q+k+1}} [C{_k}{^{q+k+1}}] \, ,
		\end{equation}
		with
		\begin{equation}
			\label{basicevaluation}
			A{^k}{_{q+k+1}} [C{_k}{^{q+k+1}}] = \frac{1}{(k+1)!} \sum\limits_{\alpha_0, \hdots, \alpha_k} \, A{^k}{_{q+k+1,\alpha_0 \hdots \alpha_k}} [C{_k}{^{q+1+k ,\alpha_0 \hdots \alpha_k}}]\, .
		\end{equation}
	\end{definition}
	
	Let us recall that the local currents and forms appearing in \eqref{basicevaluation} are all antisymmetric in their intersection indices $\alpha_0, \hdots, \alpha_k$, whence the factor $1/(k+1)!$, as in \eqref{subintegrals}.
	
	\begin{exercise}
		Show that the evaluation of gauge currents on dual gauge fields passes to classes, thus yielding the following property. In particular, show why the requirement that $i_n G{_{p+1}}{^n}$ is a \v{C}ech $(p+1)$-chain of $\mathcal{U}_M$ in the definition of a dual gauge field transformation is necessary.
	\end{exercise}
	
	\begin{property} 
		If $\boldsymbol{A}_{[q]}$ is the class of the gauge $q$-current $A_{[q]}$ and $\boldsymbol{C}^{\{q\}}$ the class of the dual gauge $q$-field $C^{\{q\}}$, then
		\begin{equation}
			\boldsymbol{A}_{[q]} \eval{ \boldsymbol{C}^{\{q\}} }  \stackrel{\mathbb{R}/\mathbb{Z}}{=} A{^k}{_{q+1+k}} [C{_k}{^{q+1+k}}] \, ,
		\end{equation}
		is a well-defined $\mathbb{R}/\mathbb{Z}$-valued quantity that we refer to as the \textbf{quantum evaluation} of $\boldsymbol{A}_{[q]}$ on $\boldsymbol{C}^{\{q\}}$.
	\end{property}
	
	\noindent 
	
	\subsection{Dual quantum $(-1)$-fields}
	
	The notion of quantum complement was originally introduced for integers which take values in $\{0, \hdots , n-1\}$. Furthermore, the quantum complement of $n$ is $-1$ and we have already mentioned that it is possible to set $H_D^{-1}(M) = \mathbb{Z}$, this $\mathbb{Z}$-module thus still sitting in some exact sequences of type \eqref{short2} and \eqref{short3}. The same kind of extension is possible for both quantum currents and dual gauge fields according to
	\begin{equation}
		H_{-1}^D(M) = \mathbb{R}/\mathbb{Z} \quad \text{and} \quad D_H^{-1} = \mathbb{Z} \, .
	\end{equation}
	Accordingly, we have the following important property.
	\begin{property}
		There exists a dual quantum field $\boldsymbol{1}^{\{-1\}}$ such that
		\begin{equation}
			\bar{\partial} \boldsymbol{1}^{\{-1\}} = d_{-1} 1\,\mbox{ and }\, cl(\boldsymbol{1}^{\{-1\}}) = 1 \, ,
		\end{equation}
		any dual quantum $(-1)$-field thus being of the form 
		\begin{equation}
			\boldsymbol{C}^{\{-1\}} = N \boldsymbol{1}^{\{-1\}} \, ,
		\end{equation}
		with $N$ the integer so that $\bar{\partial} \boldsymbol{C}^{\{-1\}} = d_{-1} N$ and $cl (\boldsymbol{C}^{\{-1\}}) = N$.
	\end{property}
	
	\begin{proof}
		Let $\mu{_0}^{0}$ be a partition of unity subordinate to $\mathcal{U}_M$. This means that, for each $U_\alpha$ of $\mathcal{U}_M$, we have a smooth function $\mu^\alpha$ with compact support in $U_\alpha$ such that 
		\begin{equation}
			\label{summualpha}
			\partial \mu{_0}^{0} = \sum_\alpha \mu^\alpha = 1 \, .
		\end{equation} 
		Then, for $k \in \{ 1 , \hdots , n-1 \}$ we consider the $k$-forms 
		\begin{equation}
			\label{defunitdualfield}
			\mu{_k}{^{k,\alpha_0 \hdots \alpha_k}} = (-1)^k \sum_{\sigma \in S_{k+1}} (-1)^{|\sigma|} \, \mu^{\alpha_{\sigma(0)}} d \mu^{\alpha_{\sigma(1)}} \wedge \hdots \wedge d \mu^{\alpha_{\sigma(k)}} \, ,
		\end{equation}
		which has compact support in each non-empty $U_{\alpha_0 \hdots \alpha_k}$, $S_{k+1}$ denoting the group of permutations of $\{0, \dots , k \}$. This gives rise to the collections $\mu{_k}{^{k}}$. These collections quite obviously fulfill the following descent equations
		\begin{equation}
			\left\{ \begin{gathered}
				d \mu{_{k-1}}{^{k-1}} = \hat{\partial} \mu{_k}{^k}\, , \hfill \\
				d_{-1} 1 = \partial \mu{_0}{^0} =  \hat{\partial}\mu{_0}{^0} \, . \hfill    
			\end{gathered} \right.
		\end{equation}
		the second equation being nothing but \eqref{summualpha}. Now, we set
		\begin{equation}
			\nu{_n}{^{n,\alpha_0 \hdots \alpha_n}} = (-1)^n \sum_{\sigma \in S_{n+1}} (-1)^{|\sigma|} \, \mu^{\alpha_{\sigma(0)}} d \mu^{\alpha_{\sigma(1)}} \wedge \hdots \wedge d \mu^{\alpha_{\sigma(n)}} \, ,
		\end{equation}
		the collection of $n$-forms thus defined fulfilling
		\begin{equation}
			d \mu{_{n-1}}{^{n-1}} = \hat{\partial} \nu{_n}{^n} \, .
		\end{equation}
		The collection $\nu{_n}{^n}$ then fulfills
		\begin{equation}
			\hat{\partial} (i_n \nu{_n}{^n}) = \oint_M d \mu{_{n-1}}{^{n-1}} = 0 \, ,
		\end{equation}
		which means that $i_n \nu{_n}{^n}$ is a real \v{C}ech $n$-cycle of $\mathcal{U}_M$. Now, for any $\omega \in \Omega_\mathbb{Z}^n(M)$, and any \v{C}ech-de Rham descent of $\omega$ whose descent collections are $\omega^{(k,p-k)}$, we have
		\begin{align}
			\nonumber
			\oint_M \omega 
			&= \sum_\alpha \oint_M \mu{_0}{^{0,\alpha}} \wedge (\delta_{-1} \omega)_\alpha \\
			&= \sum_\alpha \oint_M \mu^\alpha d \omega^{(0,n-1)}_\alpha \nonumber \\
			&= - \sum_\alpha \oint_M (d \mu^\alpha) \wedge \omega^{(0,n-1)}_\alpha \nonumber \\
			&= \sum_\alpha \oint_M (\hat{\partial} \mu{_1}{^1})^\alpha \wedge \omega^{(0,n-1)}_\alpha \nonumber \\
			\oint_M \omega
			&= \frac{1}{2} \sum_{\alpha,\beta} \oint_M \mu{_1}{^{1,\alpha \beta}} \wedge (\delta \omega^{(0,n-1)})_{\alpha \beta} \nonumber
		\end{align}
		By repeating this procedure we obtain
		\begin{equation}
			\oint_M \omega = \frac{1}{n!} \sum_{\alpha_0, \dots , \alpha_n} \oint_M \mu{_n}{^{n,\alpha_0, \dots , \alpha_n}} \wedge (\delta \omega^{(n,0)})_{\alpha_0, \dots , \alpha_n} \, ,
		\end{equation}
		which yields
		\begin{align}
			\oint_M \omega &= \frac{1}{n!} \sum_{\alpha_0, \dots , \alpha_n} \oint_M \mu{_n}{^{n,\alpha_0, \dots , \alpha_n}} \wedge (d_{-1} u^{(n,-1)})_{\alpha_0, \dots , \alpha_n} \nonumber \\
			&= u^{(n,-1)} [ i_n\nu{_n}{^n} ] \in  \mathbb{Z} \, ,
		\end{align}
		where $u^{(n,-1)}$ is the \v{C}ech $n$-cocycle ending the \v{C}ech-de Rham descent of $\omega$ \cite{We52}. So, up to a real boundary, $i_n \nu{_n}{^n}$ is a \v{C}ech $n$-cycle of $\mathcal{U}_M$. In other words, there exists a real \v{C}ech $(n+1)$-chain $t_{(n+1,-1)}$ and a \v{C}ech $n$-cycle $m_{(n,-1)}$ such that
		\begin{equation}
			i_n \nu{_n}{^n} = m_{(n,-1)} + \partial  t_{(n+1,-1)} \, .
		\end{equation}
		For each non-empty $U_{\alpha_0 \hdots \alpha_{n+1}}$ we can find an $n$-form $\rho{_{n+1}}{^{n,\alpha_0 \hdots \alpha_{n+1}}}$ has compact support in $U_{\alpha_0 \hdots \alpha_{n+1}}$ and such that
		\begin{equation}
			t_{(n+1,-1)}^{\alpha_0 \hdots \alpha_{n+1}} = \oint_M \rho{_{n+1}}{^{n,\alpha_0 \hdots \alpha_{n+1}}} \, .
		\end{equation}
		Now, if we set
		\begin{equation}
			\mu{_n}{^n} = \nu{_n}{^n} - \hat{\partial} \rho{_{n+1}}{^n} \,
		\end{equation}
		then the collection thus defined satisfies
		\begin{equation}
			\left\{\begin{gathered}
				\hat{\partial} \mu{_n}{^n} = d \mu{_{n-1}}{^{n-1}} \, , \hfill \\
				i_n\mu{_n}{^n} = m_{(n,-1)} \, . \hfill
			\end{gathered} \right.
		\end{equation}
		Hence, the $(n+2)$-tuple
		\begin{equation}
			\label{descentmuproperty}
			\mu^{\{-1\}} = \left( \mu{_0}{^0} , \hdots , \mu{_n}{^n} , m_{(n,-1)}  \right) \, , 
		\end{equation}
		has the property of a dual gauge field. With our conventions, it is a dual gauge $(-1)$-field, the dual quantum field of which is denoted by $\boldsymbol{1}^{\{-1\}}$. By definition of a partition of unity, we have $\bar{\partial}\boldsymbol{1}^{\{-1\}} = d_{-1} 1$.
		
		Let us consider $\omega \in \Omega_\mathbb{Z}^n(M)$ together with one of its \v{C}ech-de Rham descents whose components are denoted by $\omega^{(k,n-k-1)}$ and the corresponding \v{C}ech $n$-cocycle $u^{(n,-1)}$. From relation \eqref{descentmuproperty} we deduce that if $\omega$ is a normalised volume form of $M$ then $u^{(n,-1)} [m_{(n,-1)}] = 1$ which implies that the homology class of $m_{(n,-1)}$ is $1 \in H_n(M)$. In other words, $cl(\boldsymbol{1}^{\{-1\}})$ is the homology class of $M$ itself, and hence $cl(\boldsymbol{1}^{\{-1\}}) = 1$. 
		
		The fact that $\boldsymbol{C}^{\{-1\}} = N \boldsymbol{1}^{\{-1\}}$ for some integer $N$ is then a straightforward consequence of the previous result.
		
		Let $\tilde{\mu}^{\{-1\}}$ be the dual gauge $(-1)$-field defined by another partition of unitity subordinate to $\mathcal{U}_M$. We set $\Delta \mu = \tilde{\mu}^{\{-1\}} - \mu^{\{-1\}}$. The last component of $\Delta \mu$ is the difference of two \v{C}ech cocycle whose class if $1$. Hence, this last component if a \v{C}ech boundary $\partial g_{(n+1,-1)}$. There always exists a collection $G{_{n+1}}{^n}$ whose elements have local compact support such that $g_{(n+1,-1)} = i_n G{_{n+1}}{^n}$. Then, $\Delta \mu$ is equivalent to
		\begin{equation}
			\left( \Delta \mu{_0}{^0} , \dots , \Delta \mu{_{n-1}}{^{n-1}} - \partial G{_{n+1}}{^n} , 0 \right) \, .
		\end{equation}
		The penultimate component of this dual gauge field is then a collection whose elements are exact forms ( closed forms in contractible open sets) so that this dual gauge field can be written as
		\begin{equation}
			\left( \Delta \mu{_0}{^0} , \dots , d G{_{n}}{^{n-1}} , 0 \right) \, .
		\end{equation}
		The penultimate component of this dual gauge field can then be eliminated with the help of a dual gauge transformation. By repeating this procedure, we end up with a dual gauge field which is equivalent to $\Delta \mu$ and which is of the form
		\begin{equation}
			\left( \Delta \mu{_0}{^0} \pm \hat{\partial} G{_1}{^0} , 0 ,  \dots , 0 \right) \, .
		\end{equation}
		Due to the descent equations dual gauge fields must satisfy, we deduce that $\Delta \mu{_0}{^0} \pm \hat{\partial} G{_1}{^0}$ is a collection of local closed functions which must have local compact support and, as such, are all zero. Hence, $\Delta \mu{_0}{^0} \pm \hat{\partial} G{_1}{^0} = 0$, and we conclude that $\Delta \mu$ is equivalent to zero and hence that $\boldsymbol{\tilde{\mu}}^{\{-1\}} = \boldsymbol{1}^{\{-1\}}$, as claimed.
	\end{proof}
	
	Similarly, the previous exact sequences can be extended to the case where $p = n$ with
	\begin{equation}
		H_D^n(M) \simeq \mathbb{R}/\mathbb{Z} \quad , \quad H_n^D(M) \simeq \mathbb{Z} \quad , \quad D_H^n(M) \simeq \mathbb{R}/\mathbb{Z} \, .
	\end{equation}
	We have previously excluded $-1$ and $n$ in the range of quantum complement integers because these particular degrees do not correspond to quantum fields (or currents) of the BF theory. For instance, the smooth Lagrangian $\boldsymbol{B}^{[q]} \star \boldsymbol{A}^{[p]}$ is not a field of the $\U1$ BF theory although it belongs to $H_D^n(M)$. Nevertheless, extending the notion of quantum complement to the range $\{-1, \hdots , n\}$ will turn out to be relevant when trying to extend the DB product between quantum fields to a DB product between quantum currents and quantum fields as we will see it now.
	
	\subsection{Extended DB product}
	
	The most natural way to extend the DB product between quantum fields to a product between quantum currents and quantum fields is to mimic relation \eqref{productcurrentform} which defines the exterior product of a de Rham current with a form. Thus, if $\boldsymbol{B}_{[b]}$ is a quantum $b$-current and $\boldsymbol{A}^{[a]}$ a quantum $a$-field, $\boldsymbol{B}_{[b]} \star \boldsymbol{A}^{[a]}$ denotes the quantum $c$-current fulfilling
	\begin{equation}
		\label{quantumcurrentDBfields}
		\left( \boldsymbol{B}_{[b]} \star \boldsymbol{A}^{[a]} \right) \eval{\boldsymbol{C}^{\{c\}}} = \boldsymbol{B}_{[b]} \eval{\boldsymbol{A}^{[a]} \DBc \boldsymbol{C}^{\{c\}}} \, ,
	\end{equation}
	for any dual quantum $c$-field $\boldsymbol{C}^{\{c\}}$. Of course, when $\boldsymbol{B}_{[b]}$ derives from a quantum $(n-b-1)$-field $\boldsymbol{B}^{[n-b-1]}$ we want the extension $\boldsymbol{B}_{[b]} \star \boldsymbol{A}^{[a]}$ to coincide with $\boldsymbol{B}^{[n-b-1]} \star \boldsymbol{A}^{[a]}$ which is a quantum $(n-b+a)$-field and hence a quantum $(b-a-1)$-current. Thus, we have $c=b-a-1$ in relation \ref{quantumcurrentDBfields}.
	
	Since we want to extend the DB product in order to obtain an extension of the $\U1$ BF Lagrangian $\boldsymbol{B}^{[q]} \star \boldsymbol{A}^{[p]}$, $p$ and $q$ are quantum complement integers and relation \eqref{quantumcurrentDBfields} takes the form
	\begin{equation}
		\label{currentDBfieldstop}
		\left( \boldsymbol{B}_{[p]} \star \boldsymbol{A}^{[p]} \right) \eval{\boldsymbol{C}^{\{-1\}}} = \boldsymbol{B}_{[p]} \eval{\boldsymbol{A}^{[p]} \DBc \boldsymbol{C}^{\{-1\}}} \, .
	\end{equation}
	In the previous subsection, we saw that there exists an integer $n$ such that $\boldsymbol{C}^{\{-1\}} = n \boldsymbol{1}^{\{-1\}}$, with $\boldsymbol{1}^{\{-1\}}$ the generator of the free $\mathbb{Z}$-module $D_H^{-1}(M)$. Eventually, we have the defining relation
	\begin{equation}
		\left( \boldsymbol{B}_{[p]} \star \boldsymbol{A}^{[p]} \right) \eval{\boldsymbol{1}^{\{-1\}}} = \boldsymbol{B}_{[p]} \eval{\boldsymbol{A}^{[p]} \DBc \boldsymbol{1}^{\{-1\}}} \, ,
	\end{equation}
	which requires to define $\boldsymbol{A}^{[p]} \DBc \boldsymbol{1}^{\{-1\}}$. The product $\DBc$ is not the usual DB product, and is therefore not defined by \eqref{starproduct}. In fact, relation \eqref{quantumcurrentDBfields} looks very much like the one which defines the cap product. Indeed, if $a^{(k)}$ is a \v{C}ech $k$-cochain and $m_{(k+l)}$ is a \v{C}ech $(k+l)$-chain, then the cap product of $a^{(k)}$ with $m_{(k+l)}$ is the \v{C}ech $l$-chain $a^{(k)} \smallfrown m_{(k+l)}$ such that
	\begin{equation}
		(b^{(l)} \smallsmile a^{(k)}) [ m_{(k+l)} ] = b^{(l)} [ a^{(k)} \smallfrown m_{(k+l)} ] \, ,
	\end{equation}
	for any $l$-cochain $b^{(l)}$, this relation extending to classes. In fact, we know that we can associate with $\boldsymbol{A}^{[p]}$, resp. $\boldsymbol{1}^{\{-1\}}$, an element of $H_{p+1}(M)$, resp. $H_n(M)$. Hence, it is natural so imagine that the operation $\DBc$ is related to the cap product. This is why we refer to $\DBc$ as the \textbf{DB-cap product}. 
	
	The elements of the collection defining the $l$-chain $a^{(k)} \smallfrown m_{(k+l)}$ are
	\begin{equation}
		(a^{(k)} \smallfrown m_{(k+l)})^{\alpha_0 \hdots \alpha_{l}} = \frac{(l+1)!}{(k+l+1)!} \sum_{\beta_1 \hdots \beta_{k}} a^{(k)}_{\alpha_l \beta_1 \hdots \beta_k} \, m_{(k+l)}^{\alpha_0 \hdots \alpha_l \beta_1 \hdots \beta_k} \, ,
	\end{equation}
	and for $\omega^{(k,r)}$ a collection (subordinate to $\mathcal{U}_M$) of $r$-forms and $\rho{_{k+l}}{^s}$ a collection (subordinate to $\mathcal{U}_M$) of $s$-forms the \v{C}ech cap product can be extended according to
	\begin{equation}
		(\omega^{(k,r)} \Wc \rho{_{k+l}}{^s})^{\alpha_0 \hdots \alpha_{l}} =  \frac{(l+1)!}{(k+l+1)!} \sum_{\beta_1 \hdots \beta_k} \omega^{(k,r)}_{\alpha_l \beta_1 \hdots \beta_k} \wedge \rho{_{k+l}}{^{s,\alpha_0 \hdots \alpha_l \beta_1 \hdots \beta_k}} \, .
	\end{equation}
	
	The DB-cap product can now be defined as follows.
	
	\begin{definition}
		Let $A^{[p]} = \left( A^{(0,p)} , A^{(1,p-1)} , \hdots , A^{(p,0)} , a^{(p+1,-1)} \right)$ be a gauge $p$-field. The DB-cap product of $A^{[p]}$ with a dual gauge $(-1)$-field $\mu^{\{-1\}}$ deriving from a partition of unity of $M$ subordinate to $\mathcal{U}_M$ is the dual gauge $p$-field $A^{[p]} \DBc \mu^{\{-1\}}$ the components of which are
		\begin{align}
			\left\{ \begin{aligned}
				&(A^{[p]} \DBc \mu^{\{-1\}}){_k}{^{p+1+k}} 
				= d A^{(0,p)} \Wc \mu{_k}{^k} \quad \mbox{for $k \in \{0, \hdots , q-1 \}$} \, , \\
				&(A^{[p]} \DBc \mu^{\{-1\}}){_q}{^n} 
				= d A^{(0,p)} \Wc \mu{_q}{^q} 
				+ \partial Q{_{q+1}}{^n} \, ,\\
				&(A^{[p]} \DBc \mu^{\{-1\}})_{(q,-1)} = a^{(p+1,-1)} \smallfrown m_{(n,-1)} \, ,
			\end{aligned} \right.
		\end{align}
		with $q$ the quantum complement of $p$ and
		\begin{equation}
			Q{_{q+1}}{^n} = \sum_{l=0}^{p} (-1)^{l+1} A^{(l,p-l)} \Wc \mu{_{l+q+1}}{^{l+q+1}} \, .
		\end{equation}
	\end{definition}
	
	\begin{exercise}
		Show that $A^{[p]} \DBc \mu^{\{-1\}}$ is a dual gauge $p$-field of $\mathcal{U}_M$.
	\end{exercise}
	
	As already mentioned, if $\boldsymbol{B}_{[b]}$ is regular, i.e., if it derives from a quantum $(n-b-1)$-field $\boldsymbol{B}^{[n-b-1]}$, then $\boldsymbol{B}_{[b]} \star \boldsymbol{A}^{[a]}$ must coincide with $\boldsymbol{B}^{[n-b-1]} \star \boldsymbol{A}^{[a]}$. Thus, the extended DB product can also be defined by simply extending \eqref{starproduct}.
	
	\begin{definition}\label{DBpairing}
		There is a natural $\mathbb{R}/\mathbb{Z}$-valued pairing
		\begin{equation}
			\star : H^D_p(M) \times H_D^p(M) \to H^D_{-1}(M) \, ,
		\end{equation}
		which can be realised at the level of gauge representatives in the following way. Let $B_{[p]} = \left(  B{^0}{_{p+1}} , \hdots , B{^{q}}{_n} , b^{(n-p,-1)} \right)$ be a gauge $p$-current and $A^{[p]} = \left( A^{(0,p)} , A^{(1,p-1)} , \hdots , A^{(p,0)} , a^{(p+1,-1)} \right)$ a gauge $p$-field. Then, the components of the gauge $(-1)$-current $B_{[p]} \star A^{[p]}$ are
		\begin{align}
			\label{extendedstarproduct}
			\nonumber
			\left( B_{[p]} \star A^{[p]} \right)&^{k}{_{k,\alpha_0 \hdots \alpha_k}} \hfill \\
			&= \left\{ \begin{gathered}
				B{^k}{_{p+1+k,\alpha_0 \hdots \alpha_k}} \wedge dA_{\alpha_k}^{(0,p)} \hspace{0.2cm} \mbox{ for } \, 0 \leq k \leq q = n-p-1 \, , \hfill \\
				b_{\alpha_0 \hdots \;\alpha_{n - p}}^{(n-p,-1)} A_{\alpha_{n-p} \hdots \;\alpha_k}^{(k-n+p,n-k)} \hspace{0.35cm} \mbox{ for } \, n - p = q+1 \leq k \leq n\, , \hfill \\
				b_{\alpha_0 \hdots \alpha_{n-p}}^{(n-p,-1)} a_{\alpha_{n-p} \hdots \alpha_{n+1}}^{\left( p+1, -1 \right)} \hspace{0.72cm}\mbox{ for } \, k = n + 1 \, , \hfill \\ 
			\end{gathered}  \right.
		\end{align}
		where $q$ is the quantum complement of $p$.
	\end{definition}
	
	The components of $B_{[p]} \star A^{[p]}$ as defined by \eqref{extendedstarproduct} are quite obviously cup products, which is consistent with the fact that the components of $A^{[p]} \DBc \mu^{\{-1\}}$ are cap products.
	
	\begin{exercise}
		Show that	
		\begin{equation}
			\left( B_{[p]} \star A^{[p]} \right) [ \mu^{\{-1\}} ] = B_{[p]} [ A^{[p]}  \DBc \mu^{\{-1\}} ] \, .
		\end{equation}
	\end{exercise}
	
	\noindent
	
	The following property provides another way to define $B_{[p]} \star A^{[p]}$ when $p$ and $q$ are quantum complement integers.
	
	\begin{property}
		Let $B_{[p]}$ and $A^{[p]}$ be a gauge $p$-current and a gauge $p$-field, respectively. Then, the gauge $(-1)$-current $B_{[p]} \star A^{[p]}$ is gauge equivalent to a gauge $(-1)$-current of the form
		\begin{equation}
			\label{reducedDBprocut}
			\left( 0 , \hdots , d_{-1}^\dagger r^{(n,-1)}, 0 \right) \, ,
		\end{equation}
		with $r^{(n,-1)}$ a real $n$-cocycle of $\mathcal{U}_M$. The quantum $(-1)$-current $\boldsymbol{B}_{[p]} \star \boldsymbol{A}^{[p]}$ is canonically associated with an element of $H^n(M,\mathbb{R})/H^n(M) \simeq \mathbb{R}/\mathbb{Z}$.
	\end{property}
	
	\begin{proof}
		First of all, let us recall that $H_{-1}^D(M) \simeq H^n(M,\mathbb{R}/\mathbb{Z}) \simeq \mathbb{R}/\mathbb{Z}$. Hence, there always exists an element of $H^n(M,\mathbb{R}/\mathbb{Z})$ which is associated with the quantum $(-1)$-current $\boldsymbol{B}_{[p]} \star \boldsymbol{A}^{[p]}$, the gauge $(-1)$-current \eqref{reducedDBprocut} thus being a representative of this element of $H^n(M,\mathbb{R}/\mathbb{Z})$. Now, let us consider the $(-1)$-current $B_{[p]} \star A^{[p]}$ as defined by relations \eqref{extendedstarproduct}. The first component of this gauge current is the collection $B^0_{p+1} \wedge d A^{(0,p)}$, whose elements are de Rham $0$-currents, one in each non-empty $U_\alpha$. For dimensional reason, these local de Rham $0$-currents are $d^\dagger$-closed, and since each $U_\alpha$ is contractible, by applying Poincaré lemma, we conclude that these local currents are all $d^\dagger$-exact. Thus, there exists $H^0_1$ such that $B^0_{p+1} \wedge dA^{(0,p)} = d^\dagger H^0_1$. By adding to $B_{[p]} \star A^{[p]}$ the gauge current transformation $\left(-d^\dagger H^0_1, - \hat{\delta} H^0_1 ,0 , \hdots , 0 \right)$, we obtain an equivalent gauge $(-1)$-current whose first component is zero. Now, due to the descent equations, the elements forming the second component of this new gauge $(-1)$-current are necessarily $d^\dagger$-closed de Rham $1$-currents in the non-empty intersections $U_{\alpha \beta}$. As these intersections are also contractible, we deduce that all these closed currents are actually $d^\dagger$-exact. By adding to the new gauge $(-1)$-current an obvious gauge current transformation, we obtain a second new gauge $(-1)$-current, whose first two components are zero. If we go on this way, we eventually obtain a gauge $(-1)$-current of the form $\left(0 , \hdots , R^n_n , b^{(n-p,-1)} \smallsmile a^{(p+1,-1)} \right)$. Once more, the descent equations imply that $d^\dagger R^n_n = 0$. However, since we are now with local de Rham $n$-currents, this constraint only means that there exists a collection $s^{(n,-1)}$ of real numbers such that $R^n_n = d_{-1}^\dagger s^{(n,-1)}$. Thus, our previous gauge $(-1)$-current reads $\left(0 , \hdots , d_{-1}^\dagger s^{(n,-1)} , b^{(n-p,-1)} \smallsmile a^{(p+1,-1)} \right)$. In a final effort, we can notice that, for dimensional reasons, the $(n+1)$-cocycle $b^{(n-p,-1)} \smallsmile a^{(p+1,-1)}$ is necessarily a coboundary, i.e., $b^{(n-p,-1)} \smallsmile a^{(p+1,-1)} = \delta u^{(n,-1)}$ for some \v{C}ech $n$-cochain $u^{(n,-1)}$. A last obvious gauge transformation finally yields the gauge $(-1)$-current $\left(0 , \hdots , d_{-1}^\dagger s^{(n,-1)} - \delta u^{(n,-1)} , 0 \right)$, which is of the form \eqref{reducedDBprocut} if we set $r^{(n,-1)} = d_{-1}^\dagger s^{(n,-1)} - \delta u^{(n,-1)}$. Moreover, $\delta r^{(n,-1)} = 0$ by the descent equations. The last statement of the property is quite obvious.
	\end{proof}
	
	Of course, the above property applies to the DB product of quantum fields. In fact, it even extends beyond the abelian framework. For instance, the $\SU2$ Chern-Simons Lagrangian is also a quantum $(-1)$-current (or a quantum $n$-field if smooth fields are used) \cite{CS73}, even if it is not quadratic in the fields as in the abelian case. In any case, the dependence of the real $n$-cocycle $r^{(n,-1)}$ on fields, whether abelian or not, is unfortunately only implicit. The interest of the above property is that evaluation of $B_{[p]} \star A^{[p]}$ on a dual gauge $(-1)$-field $C^{\{-1\}}$ will reduce to the evaluation of the real cocycle $r^{(n,-1)}$ on the (integer) cycle associated with $C^{\{-1\}}$. The drawback of this approach is that $r^{(n,-1)}$ depends on $A^{[p]}$ and $B_{[p]}$ in a very cryptic way (which is even more cryptic in the non-abelian CS case).
	
	Eventually, the evaluation of the DB product of a quantum $p$-current with a quantum $p$-field can be obtained as follows.
	
	\begin{property}
		Let $\boldsymbol{B}_{[p]}$ be a quantum $p$-current and $\boldsymbol{A}^{[p]}$ a quantum $p$-field. Let $r^{(n,-1)}$ be a real \v{C}ech $n$-cocycle of $\mathcal{U}_M$ such that the gauge $(-1)$-current $R_{[-1]} = \left( 0 , \hdots , d_{-1}^\dagger r^{(n,-1)}, 0 \right)$ is a representative of $\boldsymbol{B}_{[p]} \star \boldsymbol{A}^{[p]}$, and let $M_{\{n\}} = \left( M_{(0,n)} , \hdots , M_{(n,0)} , m_{(n,-1)} \right)$ be a decomposition of $M$ subordinate to $\mathcal{U}_M$. Then, we have
		\begin{equation}
			\label{evalDBon1}
			(\boldsymbol{B}_{[p]} \star \boldsymbol{A}^{[p]}) \eval{\boldsymbol{1}^{\{-1\}}} \stackrel{\mathbb{R}/\mathbb{Z}}{=} (-1)^n r^{(n,-1)} \left[ m_{(n,-1)} \right] \, ,
		\end{equation} 
	\end{property}
	
	\begin{exercise}
		Show relation \eqref{evalDBon1}.
	\end{exercise}
	
	Let us point out that the extension of the DB product to a product between quantum $p$-current and quantum $p$-fields is well-defined thanks to the use of a partition of unity subordinate to $\mathcal{U}_M$, such a partition of unity inducing the dual quantum $(-1)$-field $\boldsymbol{1}^{\{-1\}}$ which generates $D_H^{-1}(M) = \mathbb{Z}$. In fact, the representative of $\boldsymbol{1}^{\{-1\}}$ induced by such a partition of unity can be seen as a ``smooth decomposition" of $M$ subordinate to $\mathcal{U}_M$. If we try to use a polyhedral decomposition when defining the evaluation $(\boldsymbol{B}_{[p]} \star \boldsymbol{A}^{[p]}) \eval{\boldsymbol{1}^{\{-1\}}}$, we would immediately face the problem of defining the exterior product of local de Rham currents, those coming from the local representative of $\boldsymbol{B}_{[p]}$ and those coming from the polyhedral decomposition. The use of a smooth decomposition inferred by a partition of unity subordinate to $\mathcal{U}_M$ prevents these possible issues. Nevertheless, there are many cases where the product of de Rham currents, i.e., the product of distributions, is well-defined \cite{dR55}. In analogy with these cases, we can say that the use of a smooth decomposition provides a regularization which gives a meaning to the evaluation $(\boldsymbol{B}_{[p]} \star \boldsymbol{A}^{[p]}) \eval{\boldsymbol{1}^{\{-1\}}}$, even if we represent $\boldsymbol{1}^{\{-1\}}$ with a polyhedral decomposition of $M$. In fact, this use of a smooth decomposition of $M$ provides an idea of how to extend the DB product to a product of quantum currents of quantum complement dimensions, of course in some specific cases and not in general. This will be discussed in the next section. When there is no risk of confusion, we denote by $\boldsymbol{1}$ instead of $\boldsymbol{1}^{\{-1\}}$ the generator of $D_H^{-1}(M)$. 
	
	\vspace{0.5cm}
	
	Before discussing Pontryagin duality, let us give here a reminder about the terminology and the notations we met so far.
	\begin{itemize}
		
		\item[-] If $p \in \left\{ 0 , \hdots , n-1 \right\}$, then $q=n-p-1$ is the quantum complement of $p$ in $M$;
		
		\item[-] A classical $p$-field of $M$, i.e., a $p$-form, is generically denoted by $\chi^{(p)}$, its class in $\Omega^p(M)/\Omega_\mathbb{Z}^p(M)$ by $\bar{\chi}^{(p)}$;
		
		\item[-] A classical $p$-current of $M$, i.e., a de Rham $q$-current, is generically denoted by $\chi_{(q)}$, and its class in $\Omega_q(M)/\Omega^\mathbb{Z}_q(M)$ by $\bar{\chi}_{(q)}$;
		
		\item[-] A gauge $p$-field of $M$ is generically denoted by $A^{[p]}$, the gauge $p$-field defined by a classical $p$-field $\chi^{(p)}$ being then denoted by $\chi^{[p]}$;
		
		\item[-] A gauge $q$-current of $M$ is generically denoted by $A_{[q]}$, the gauge $q$-current defined by a classical $q$-current $\chi_{(q)}$ being then denoted by $\chi_{[q]}$;
		
		\item[-] A gauge $p$-field transformation is generically denoted by $G^{[p]}$;
		
		\item[-] A gauge $q$-current transformation is generically denoted by $G_{[p]}$;
		
		\item[-] A quantum $p$-field of $M$ is generically denoted by $\boldsymbol{A}^{[p]}$ and the $\mathbb{Z}$-module of quantum $p$-fields by $H_D^p(M)$; the quantum $p$-field defined by a classical $p$-field $\chi^{(p)}$ and which is the image of $\bar{\chi}^{(p)}$ in $H_D^p(M)$ is then denoted by $\boldsymbol{\chi}^{[p]}$;		
		
		\item[-] A quantum $q$-current of $M$ is generically denoted by $\boldsymbol{A}_{[q]}$ and the $\mathbb{Z}$-module of quantum $q$-currents by $H^D_q(M)$; the quantum $q$-current defined by a classical $q$-current $\chi_{(q)}$ and which is the image of $\bar{\chi}_{(q)}$ in $H^D_q(M)$ is then denoted by $\boldsymbol{\chi}_{[q]}$;
		
		\item[-] The set of quantum $(-1)$-fields, resp. $(-1)$-currents, is identified with $\mathbb{Z}$, resp. $\mathbb{R}/\mathbb{Z}$; 
		
		\item[-] A dual gauge $p$-field is generically denoted by $C^{\{p\}}$;
		
		\item[-] A dual gauge $p$-field transformation is generically denoted by $G^{\{p\}}$; 
		
		\item[-] A dual quantum $p$-field is generically denoted by $\boldsymbol{C}^{\{p\}}$ and the $\mathbb{Z}$-module of dual quantum $p$-fields by $D_H^p(M)$;
		
		\item[-] The set of dual quantum $(-1)$-fields is identified with $\mathbb{Z}$; it is the natural $\mathbb{Z}$-module on which quantum $(-1)$-currents are evaluated.

	\end{itemize}
	
	\section{Pontryagin duality}
	\label{section_pontryagin_duality}
	
	The Pontryagin dual of a locally compact abelian group $G$ is the group $\Hom_{\mathbb{Z}} (G,\mathbb{R}/\mathbb{Z})$ of continuous group homomorphisms from $G$ to $\mathbb{R}/\mathbb{Z}$. To lighten notations, we denote $G^\star$ the Pontryagin dual of $G$. 
	
	\subsection{Dual exact sequences}
	
	As explained in \cite{HLZ}, $H_D^p(M)$ can be endowed with a topology inherited from the $\mathcal{C}^\infty$-topology of $\Omega_\mathbb{Z}^{p+1}(M)$ and the standard topology of the torus part of $H^p(M,\mathbb{R}/\mathbb{Z})$. This topology allows us to consider $H_D^p(M)^\star$, the Pontryagin dual of $H_D^p(M)$. We have already met two subsets of $H_D^p(M)^\star$.
	
	\begin{property} Firstly, let $z_p$ be a $p$-cycle of $M$. This cycle defines an element of $H_D^p(M)^\star$, also denoted by $z_p$, according to
		\begin{equation}
			\label{injectionZintoH}
			\forall\, \boldsymbol{A}^{[p]} \in H_D^p(M) , \quad z_p \eval{\boldsymbol{A}^{[p]}} = \oint_{z_p} \boldsymbol{A}^{[p]} \, . \hspace{2,5cm}
		\end{equation}
		Thus, we have the trivial injection
		\begin{equation}
			\label{inclusion1InH}
			Z_p(M) \hookrightarrow H_D^p(M)^\star \, .
		\end{equation}
		
		\noindent Secondly, let $\boldsymbol{B}^{[q]}$ be a quantum $q$-field. This quantum $q$-field defines an element of $H_D^p(M)^\star$, denoted by $\boldsymbol{B}_{\{p\}}$, according to
		\begin{equation}
			\label{injectionHintoH'}
			\forall\,\boldsymbol{A}^{[p]} \in H_D^p(M) , \quad \boldsymbol{B}_{\{p\}} \eval{\boldsymbol{A}^{[p]}} = \oint_M \boldsymbol{B}^{[q]} \star \boldsymbol{A}^{[p]} \, . \hspace{2.5cm}
		\end{equation}
		Thus, we have the injection
		\begin{equation}
			\label{inclusion2InH}
			H_D^q(M) \hookrightarrow H_D^p(M)^\star \, .
		\end{equation}
		
		Each of the above injections is referred to as an inclusion, the latter one being the analog of the injection of the space of $p$-forms into the space of de Rham $(n-p)$-currents which itself is a generalization of the inclusion of the space of smooth functions into the space of distributions.
		
	\end{property} 
	
	Since quantum fields have gauge fields as local representatives, inclusion \eqref{inclusion2InH} highly suggests that the elements of $H_D^p(M)^\star$ might admit local representatives too. However, inclusion \eqref{inclusion1InH} also suggests that these representatives cannot be gauge currents. Indeed, the integral appearing in \eqref{injectionZintoH} is defined with the help of a decomposition $z_{\{p\}} = \left(c_{\left( 0,p \right)} , \hdots\right.$ $ \left.\hdots,c_{\left( p,0 \right)} , c_{\left( p,-1 \right)} \right)$ of $z_p$ which seems to play the role of a local representative of $z_p \in H_D^p(M)^\star$. Now, even if we see the components of $z_{\{p\}}$ as local de Rham currents, these components fulfill descent equations \eqref{decompcycle} which are obviously not those of a gauge current. Moreover, all the local de Rham currents appearing in $z_{\{p\}}$ have support constraint unlike the local currents forming the components of a gauge current. Nevertheless, the extension of the DB product introduced in the last part of the previous subsection could be interpreted as a naive way to define a $\mathbb{R}/\mathbb{Z}$-valued evaluation of quantum $p$-currents along quantum $p$-fields, thus suggesting that there exists some relation between $H^D_p(M)$ and $H_D^p(M)^\star$.
	
	By applying the contravariant functor $\Hom_{\mathbb{Z}}(\, \cdot \,,\mathbb{R}/\mathbb{Z})$ to exact sequence \eqref{short3}, we obtain the following dual exact sequence 
	\begin{equation}
		\label{sequencePont1}
		0 \rightarrow \Omega_\mathbb{Z}^{q+1}(M)^\star \xrightarrow{\bar{d}^\star} H_D^q(M)^\star \xrightarrow{i^\star} H^{p+1}(M) \rightarrow 0 \, .
	\end{equation}
	Firstly, $H^q(M,\mathbb{R}/\mathbb{Z})^\star \cong H_q(M)$ since, by the Universal Coefficient Theorem, we find $H^q(M,\mathbb{R}/\mathbb{Z}) \cong \Hom_{\mathbb{Z}} (H_q(M),\mathbb{R}/\mathbb{Z})$, and from Poincaré duality we eventually obtain $H^q(M,\mathbb{R}/\mathbb{Z})^\star \cong H^{n-q}(M) = H^{p+1}(M)$,  $p$ being the quantum complement of $q$. The homomorphisms $\bar{d}^\star$ and $i^\star$ are the Pontryagin dual\textcolor{red}{s?} of $\bar{d}$ and $i$, respectively. The map $\bar{d}^\star$ is defined by setting that 
	\begin{equation}
		\label{dualcurv}
		\bar{d}^\star(\bar{\chi}) \eval{\boldsymbol{B}^{[q]}} = \bar{\chi} [\bar{d} \boldsymbol{B}^{[q]}] = \bar{\chi} [F(\boldsymbol{B}^{[q]})] \, .
	\end{equation}
	for any $\bar{\chi} \in \Omega_\mathbb{Z}^{q+1}(M)^\star$ and any $\boldsymbol{B}^{[q]} \in H_D^q(M)$. The notation $\bar{\chi}$ stems from isomorphism \eqref{iso1} below. The morphism $\bar{d}^\star$ is injective since if $\bar{d}^\star (\bar{\chi}) = 0$ then $\bar{\chi} (F) = 0$ for any $F \in \Omega^{q+1}_\mathbb{Z}(M)$, which implies that $\bar{\chi} = 0$. Similarly, $i^\star$ is defined by setting
	\begin{equation}
		i^\star (\boldsymbol{\Phi}) [\boldsymbol{\nu}] = \boldsymbol{\Phi} \eval{i(\boldsymbol{\nu})} \, .
	\end{equation}
	for any $\boldsymbol{\Phi} \in H_D^q(M)^\star$ and any $\boldsymbol{\nu} \in H^q(M,\mathbb{R}/\mathbb{Z})$. From exact sequence \eqref{short2} we know that for any $\boldsymbol{a} \in H^{p+1}(M)$ there exists a quantum $p$-fields $\boldsymbol{A}^{[p]}$ such that $cl(\boldsymbol{A}^{[p]}) = \boldsymbol{a}$. Now, since we have inclusion \eqref{inclusion2InH}, we know that $\boldsymbol{A}^{[p]}$ defines an element of $H_D^q(M)^\star$. Hence, $i^\star$ is surjective. Finally, we have $i^\star \circ \bar{d}^\star = (\bar{d} \circ i)^\star = 0$ since $\bar{d} \circ i = 0$. This ends to explain why the above sequence is exact. The right-hand part of exact sequence \eqref{sequencePont1} looks a lot like the right-hand side of exact sequence \eqref{currentshort1}, thus consolidating the idea of a relation between $H_D^q(M)^\star$ and $H_p^D(M)$. 
	
	The Pontryagin dual of exact sequence \eqref{short2} yields the dual exact sequence
	\begin{equation}
		\label{sequencePont2}
		0 \rightarrow H^{p}(M,\mathbb{R}/\mathbb{Z}) \xrightarrow{cl^\star} H_D^q(M)^\star \xrightarrow{\bar{\delta}^\star} \left( \frac{\Omega^q(M)}{\Omega_\mathbb{Z}^q(M)} \right)^\star \rightarrow 0 \, ,
	\end{equation}
	with $H^{q+1}(M)^\star \cong H^{p}(M,\mathbb{R}/\mathbb{Z})$ since $H^q(M,\mathbb{R}/\mathbb{Z})^\star \cong H^{p+1}(M)$, $p$ and $q$ being quantum complement to each other and the morphisms $cl^\star$ and $\bar{\delta}^\star$ being the Pontryagin duals of $cl$ and $\bar{\delta}$, respectively.
	
	\begin{exercise}
		Show that sequence \eqref{sequencePont2} is exact.
	\end{exercise}
	
	In \eqref{sequencePont2} we recover some elements of exact sequence  \eqref{currentshort2} in which $H_p^D(M)$ is sitting. Before definitively clarifying the relation between $H_p^D(M)$ and $H_D^q(M)^\star$, let us exhibit a property which will serve this purpose.
	
	\begin{property}
		\begin{equation}
			\label{iso1}
			\frac{\Omega_{q+1}(M)}{\Omega_{q+1}^\mathbb{Z}(M)} \cong \Omega_\mathbb{Z}^{q+1}(M)^\star \, .
		\end{equation}
	\end{property}
	
	\begin{proof}
		The above isomorphism can be deduced from the comparison of the following exact sequence
		\begin{equation}
			\label{sequencePont3}
			0 \rightarrow (F^{q+1})^\star \xrightarrow{cl^\star} \Omega_\mathbb{Z}^{q+1}(M)^\star \xrightarrow{d^\star} \left( \frac{\Omega^q(M)}{\Omega_\circ^q(M)} \right)^\star \rightarrow 0\, ,
		\end{equation}
		itself obtained as the Pontryagin dual of
		\begin{equation}
			0 \rightarrow \left( \frac{\Omega^q(M)}{\Omega_\circ^q(M)} \right) \xrightarrow{d} \Omega_\mathbb{Z}^{q+1}(M) \xrightarrow{cl}  F^{q+1} \rightarrow 0\, ,
		\end{equation}
		with the exact sequence
		\begin{equation}
			\label{sequtopont}
			0 \to \frac{\Omega^\circ_{q+1}(M)}{\Omega^\mathbb{Z}_{q+1}(M)} \to \frac{\Omega_{q+1}(M)}{\Omega^\mathbb{Z}_{q+1}(M)} \to \frac{\Omega_{q+1}(M)}{\Omega^\circ_{q+1}(M)} \to 0 \, ,
		\end{equation}
		which is the natural extension of \eqref{short4} to de Rham currents. The comparison relies on the following points. Firstly, $\Omega^\circ_{q+1}(M)/\Omega^\mathbb{Z}_{q+1}(M) \cong (\mathbb{R}/\mathbb{Z})^{b_{q+1}} \cong (F^{q+1})^\star$. Secondly, $\Omega_{q+1}(M)/\Omega^\mathbb{Z}_{q+1}(M) \hookrightarrow \Omega_\mathbb{Z}^{q+1}(M)^\star$ in a quite obvious way. Thirdly, from $(d \Omega^q(M))^\star \cong d \Omega_{q+1}(M)$ \cite{HLZ} we deduce $(\Omega^q(M)/\Omega_\circ^q(M))^\star \cong \Omega_{q+1}(M)/\Omega^\circ_{q+1}(M)$. Finally, by applying  the short 5-lemma to \eqref{sequencePont3} and \eqref{sequtopont}, we obtain relation \eqref{iso1}.
	\end{proof} 	 
	
	\noindent The injection $cl^\star$ which appears in exact sequence \eqref{sequencePont3} can be realized in the following way. Let us select a collection $(z_{q+1}^\mu)_\mu$ of free $(q+1)$-cycles of $M$ which generate $F_{q+1} \cong F^{q+1} \cong (\mathbb{R}/\mathbb{Z})^{b_{q+1}}$, and for any $\vec{\theta} \in (\mathbb{R}/\mathbb{Z})^{b_{q+1}}$, let us consider $Z_\theta = \sum_{\mu} \theta_\mu Z_{q+1}^\mu$ where $Z_{q+1}^\mu$ denotes the de Rham current of $z_{q+1}^\mu$, and set for any $\alpha \in \Omega_\mathbb{Z}^{q+1}(M)$, $Z_\theta [\alpha] = \sum_{\mu} \theta_\mu Z_{q+1}^\mu [\alpha] \in \mathbb{R}/\mathbb{Z}$. This defines $Z_\theta$ as an element of $\Omega_\mathbb{Z}^{q+1}(M)^\star$. Now, if we replace each $Z_{q+1}^\mu$ by a cohomologous current $Z_{q+1}^\mu + d^\dagger T_{q+1}^\mu$, then $Z_\theta$ is also changed by a coboundary, so that it defines the same element of $\Omega_\mathbb{Z}^{q+1}(M)^\star$ as $Z_\theta$. In other words, $Z_\theta$ only depends on $\vec{\theta}$, the construction thus defining the injection $cl^\star$. Furthermore, $\vec{\theta}$ obviously defines an element $\vec{\theta}_{(q)}$ of $\Omega^\circ_{q+1}(M)/\Omega^\mathbb{Z}_{q+1}(M) \cong (\mathbb{R}/\mathbb{Z})^{b_{q+1}}$. The gauge $q$-current defined by $\vec{\theta}_{(q)}$ is then denoted by $\vec{\theta}_{[q]}$, and the associated quantum $q$-current by $\vec{\boldsymbol{\theta}}_{[q]}$.
	
	We can now state the theorem that confirms our previous impressions concerning the relation between quantum currents and Pontryagin duality.
	
	\begin{theorem}
		\begin{equation}
			\label{isoHstarH'}
			H_q^D(M) \cong H_D^q(M)^\star \, .
		\end{equation}
		
	\end{theorem}
	\begin{proof}
		Let us compare exact sequences \eqref{currentshort1} and \eqref{sequencePont1}. Firstly, both sequences have the cohomology group $H^{p+1}(M)$ on their last nontrivial entry. Secondly, their first nontrivial entry are isomorphic spaces, cf. \eqref{iso1}. Finally, we can associate with a quantum $q$-field $\boldsymbol{A}_{[q]}$ an element $\boldsymbol{A}_{\{q\}}$ of $H_D^q(M)^\star$ defined by
		\begin{equation}
			\label{dualcurrentsfromcurrents}
			\forall\,\boldsymbol{B}^{[q]} \in  H_D^q(M), \quad \boldsymbol{A}_{\{q\}} \eval{\boldsymbol{B}^{[q]}} = (\boldsymbol{A}_{[q]} \star \boldsymbol{B}^{[q]}) \eval{\boldsymbol{1}} \, , \hspace{1.5cm}
		\end{equation}
		this association thus defining a map $H_q^D(M) \to H_D^q(M)^\star$. Applying the 5-lemma to these sequences achieves the proof.
	\end{proof}	
	
	The isomorphism of the above theorem stands at the level of quantum currents. To get a version which deals with gauge currents, it is necessary to have local representatives for the elements of $H_D^q(M)^\star$. Before exhibiting such representatives, let us note that if we compare exact sequences \eqref{currentshort2} and \eqref{sequencePont2}, knowing that $H_q^D(M) \cong H_D^q(M)^\star$, we deduce that $(\Omega^q(M)/\Omega_\mathbb{Z}^q(M))^\star$ plays the same role as the group of curvatures, $\Omega^\mathbb{Z}_q$, these two groups being indeed isomorphic. Thus, we can consider $(\Omega^q(M)/\Omega_\mathbb{Z}^q(M))^\star$ as the group of curvatures from the point of view of Pontryagin duality.
	
	\subsection{Local representatives of dual quantum currents}
	
	As elements of $H_D^p(M)^\star$, $p$-cycles have local representatives provided by their subordinate decompositions. This can be used as a model to obtain representatives for the elements of $H_D^q(M)^\star$. Let $z_{\{p\}} = \left( c_{(0,p)} , \hdots , c_{(p,0)} , c_{(p,-1)} \right)$ be a decomposition subordinate to $\mathcal{U}_M$ of a $p$-cycle $z_p$. First, we replace each of the chains forming $c_{(k,p-k)}$ by its de Rham current in order to obtain the $(p+2)$-tuple $\left( C_{(0,p)} , \hdots , C_{(p,0)} , c_{(p,-1)} \right)$. Then, the de Rham currents forming $C_{(k,p-k)}$ in this $(p+2)$-tuple have compact support in the appropriate intersection $U_{\alpha_0 \hdots \alpha_k}$ and they fulfill the descent equation
	\begin{equation}
		\left\{ 
		\begin{gathered}
			\partial {\kern 1pt} C_{(0,p)} = j_z \, , \hfill \\
			\partial C_{(k,p-k)} = d^\dagger {\kern 1pt} C_{(k-1,p-k+1)} \, , \hfill \\ 
		\end{gathered}  \right.
	\end{equation}
	for $k \in \{1, \hdots , p\}$, where $j_z$ is the de Rham current of $z$. The \v{C}ech $p$-chain $c_{(p,-1)}$ is related to $C_{(p,0)}$ in the following way. Since $C_{(p,0)}$ is a collection of local compactly supported de Rham $0$-currents, one in each $U^z_{\alpha_0 \hdots \alpha_p}$, these currents are automatically $d^\dagger$-closed. Thus, if $X^{\alpha_0 \hdots \alpha_p}$ is a de Rham $0$-current with compact support in $U^z_{\alpha_0 \hdots \alpha_p}$ which generates the $0$-th cohomology group of compactly supported de Rham $0$-currents in this open set, then
	\begin{equation}
		C_{(p,0)}^{\alpha_0 \hdots \alpha_p} = c^{\alpha_0 \hdots \alpha_p} X^{\alpha_0 \hdots \alpha_p} + d^\dagger T^{\alpha_0 \hdots \alpha_p}
	\end{equation}
	for some $1$-current $T^{\alpha_0 \hdots \alpha_p}$. Then, we define the equivalent of $b_0$ on (compactly supported) de Rham $p$-currents
	\begin{equation}
		d^\dagger_0 C_{(p,0)}^{\alpha_0 \hdots \alpha_p} = c^{\alpha_0 \hdots \alpha_p} \, ,
	\end{equation}
	and we have 
	\begin{equation}
		c_{(p,-1)} = d^\dagger_0 C_{(p,0)} \, .
	\end{equation}
	By construction, $c_{(p,-1)}$ is a \v{C}ech $p$-cycle of $\mathcal{U}_M$. We are ready to define local representatives of elements of $H_D^q(M)^\star$.
	
	\begin{definition}
		Let $c_{(p,-1)}$ be a \v{C}ech $p$-cycle of $\mathcal{U}_M$ and let $C_{(k,p-k)}$ be a collection of local de Rham $(p-k)$-currents $C_{(k,p-k)}^{\alpha_0 \hdots \alpha_k}$ which have compact support in $U^z_{\alpha_0 \hdots \alpha_p}$. The $(p+2)$-tuple $C_{\{p\}} = \left( C_{(0,p)} , \hdots , C_{(p,0)} , c_{(p,-1)} \right)$ is called a \textbf{dual gauge $p$-current} if
		\begin{equation}
			\label{localdualcurrents}
			\left\{ 
			\begin{gathered}
				d^\dagger {\kern 1pt} C_{(k,p-k)} = \partial C_{(k+1,p-k-1)} \, , \hfill \\ 
				d^\dagger_0 C_{(p,0)} = c_{(p,-1)} \, , \hfill
			\end{gathered}  \right.
		\end{equation}
		for $k \in \{ 0 , \dots , p-1 \}$. A \textbf{dual gauge p-current transformation} is a dual gauge $p$-current of the form
		\begin{equation}
			\hspace{-0.25cm} \left( {\partial G_{(1,p)} , \partial G_{(2,p-1)} + d^\dagger G_{(1,p)}} , \hdots  , \partial G_{(p+1,0)} + d^\dagger G_{(p,1)} , \partial d^\dagger_0 G_{(p+1,0)} \right) \, ,
		\end{equation}
		where $G_{(p+1,0)}$ is such that $d^\dagger_0 G_{(p+1,0)}$ is a \v{C}ech chain of $\mathcal{U}_M$. It is generically denoted by $G_{\{p\}}$. By identifying dual gauge $p$-currents which differ by a dual gauge current transformation, we define an equivalence relation, the classes of which are called \textbf{dual quantum currents}. A dual quantum $p$-current is generically denoted by $\boldsymbol{C}_{\{p\}}$ and the $\mathbb{Z}$-module of dual quantum $p$-currents by $D^H_p(M)$.
	\end{definition}
	
	It seems logical to look for an evaluation formula for dual quantum currents to show that $D^H_p(M)$ can be canonically identified with $H_D^p(M)^\star$.
	
	\begin{definition}
		Let $C_{\{p\}}$ be a dual gauge $p$-current. Then, for any gauge $p$-field $A^{[p]}$ we set
		\begin{equation}
			C_{\{p\}} [A^{[p]}] = \sum\limits_{k = 0}^p (-1)^k C_{(k,p-k)} [ A^{(k,p-k)}] \, ,
		\end{equation}
		with
		\begin{equation}
			\label{compoevaluation}
			C_{(k,p-k)} [A^{(k,p-k)}] =  \frac{1}{(k+1)!} \sum\limits_{\alpha_0, \hdots, \alpha_k} \, C_{(k,p-k)}^{\alpha_0 \hdots \alpha _k} [ A_{\alpha_0 \hdots \alpha_k}^{(k,p-k)} ] \, .
		\end{equation}
		This is referred to as the \textbf{evaluation} of $C_{\{p\}}$ on $A^{[p]}$.
	\end{definition}
	
	Of course, when $C_{\{p\}}$ is derived from a decomposition $z_{\{p\}}$ of a $p$-cycle $z_p$ of $M$, then the evaluation of $C_{\{p\}}$ on a gauge $p$-field coincides with the integral of this gauge field along $z_{\{p\}}$. Since integration along cycles extends to quantum fields, the following property is not surprising.
	
	\begin{property}
		The evaluation of dual gauge $p$-currents on gauge $p$-fields naturally extends to a \textbf{quantum evaluation} of dual quantum $p$-currents on quantum $p$-fields by setting
		\begin{equation}
			\boldsymbol{C}_{\{p\}} \eval{\boldsymbol{A}^{[p]}}  \stackrel{\mathbb{R}/\mathbb{Z}}{=} C_{\{p\}} [ A^{[p]} ] \, ,
		\end{equation}
		$C_{\{p\}}$ and $A_{[p]}$ being representatives of the dual quantum $p$-current $\boldsymbol{C}_{\{p\}}$ and of the quantum $p$-field $\boldsymbol{A}^{[p]}$, respectively.
	\end{property}
	
	\begin{proof}
		We just need to prove that the evaluation of a dual gauge transformation on any gauge $p$-field is an integer and that the same applies to the evaluation of any dual gauge $p$-current on a gauge field transformation.
		Let $G_{\{p\}}$ be a dual gauge transformation and $A^{[p]}$ a gauge $p$-field. Then
		\begin{equation}
			G_{\{p\}} [A^{[p]}] = \sum\limits_{k = 0}^p (-1)^k \left(\partial G_{(k+1,p-k)} + d^\dagger G_{(k,p-k+1)}\right) [A^{(k,p-k)}] \, ,
		\end{equation}
		with $G_{(0,p+1)} = 0$ by definition. Since $\partial$ and $\delta$ on the one hand, and $d^\dagger$ and $d$ on the other hand, are each other's adjoint, we have 
		\begin{equation}
			\left\lbrace
			\begin{gathered}
				(\partial G_{(k+1,p-k)}) [A^{(k,p-k)}] = G_{(k+1,p-k)} [\delta A^{(k,p-k)}] \, , \hfill \\
				(d^\dagger G_{(k,p-k+1)}) [A^{(k,p-k)}] = G_{(k,p-k+1)} [d A^{(k,p-k)}] \hfill \, ,
			\end{gathered}
			\right.
		\end{equation}
		Now, taking into account the descent equations of $A^{[p]}$ and the coefficient $(-1)^k$ in the sum defining $G_{\{p\}} [A^{[p]}]$, we eventually obtain
		\begin{align}
			G_{\{p\}} [A^{[p]}] 
			\nonumber
			&= (-1)^p G_{(p+1,0)} [d_{-1} a^{(p+1,-1)}] \\ 
			&= (-1)^p (d^\dagger_0 G_{(p+1,0)}) [a^{(p+1,-1)}]\, ,
		\end{align}
		and since, by construction, $d^\dagger_0 G_{(p+1,0)}$ is a \v{C}ech $(p+1)$-chain of $\mathcal{U}_M$ and $a^{(p+1,-1)}$ a \v{C}ech $(p+1)$-cocycle of $\mathcal{U}_M$, the right-hand side of the second line in the above expression is an integer, i.e., zero in $\mathbb{R}/\mathbb{Z}$.
		The other verification is going exactly the same way and we get in that case
		\begin{align}
			C_{\{p\}} [G^{[p]}]
			\nonumber
			&= (-1)^p C_{(p,0)} [d_{-1} q^{(p,-1)}] \\ 
			&= (-1)^p (d^\dagger_0 C_{(p,0)}) [q^{(p,-1)}]\, .
		\end{align}
		This also yields an integer since by construction $d^\dagger_0 C_{(p,0)}$ is a \v{C}ech $p$-cycle  of $\mathcal{U}_M$ and $q^{(p,-1)}$ is a \v{C}ech $p$-cochain of $\mathcal{U}_M$.
	\end{proof}
	
	The relation between $p$-cycles and dual quantum $p$-currents can be extended to a quantum $q$-field. 
	
	\begin{property}
		Any gauge $q$-field of $\mathcal{U}_M$ gives rise to a dual gauge $p$-current of $\mathcal{U}_M$. This relation passes to the classes so that any quantum $q$-field defines a dual quantum $p$-currents.
	\end{property}
	
	\begin{proof}
		Let $\boldsymbol{B}^{[q]}$ be a quantum $q$-field and let $M_{\{n\}}$ be the dual gauge $n$-current defined by a decomposition of $M$ subordinate to $\mathcal{U}_M$. We are looking for a dual gauge $p$-current $B_{\{p\}}$ such that $\boldsymbol{B}_{\{p\}}  \eval{\boldsymbol{A}^{[p]}}  = \oint_M  \boldsymbol{B}^{[q]} \star \boldsymbol{A}^{[p]}$ for any quantum $p$-field $\boldsymbol{A}^{[p]}$. From what was done in the previous section, it is quite obvious that, by setting
		\begin{equation}
			B_{(k,p-k)} = (-1)^{k(q+1)} \, b^{(q+1,-1)} \smallfrown M_{(q+k+1,p-k)} \, ,
		\end{equation}
		for $k \in \{1 , \hdots , p-1 \}$, as well as
		\begin{equation}
			B_{(0,p)} = d^\dagger \Big(\sum_{k=0}^{q} (-1)^k  M_{(k,n-k)} \Wc B^{(k,q-k)}\Big) + (b^{(q+1,-1)} \smallfrown M_{(q+1,p)}) \, ,
		\end{equation}	
		we obtain a dual gauge $p$-current such that $B_{\{p\}}  [A^{[p]}]  = \oint_M  B^{[q]} \star A^{[p]}$. Indeed, for $k \in \{2, \hdots , p\}$ we have
		\begin{align}
			\partial B_{(k,p-k)} 
			&= (-1)^{k(q+1)} \partial (b^{(q+1,-1)} \smallfrown M_{(q+k+1,p-k)}) \nonumber \\
			&= (-1)^{k(q+1)} (-1)^{q+1} \big( b^{(q+1,-1)} \smallfrown \partial M_{(q+k+1,p-k)} \big) \nonumber \\
			&= (-1)^{(k+1)(q+1)} \big( b^{(q+1,-1)} \smallfrown d^\dagger M_{(q+k,p-k+1)} \big) \nonumber \\
			&= d^\dagger \big( (-1)^{(k-1)(q+1)} \, b^{(q+1,-1)} \smallfrown M_{(q+k,p-k+1)} \big) \nonumber \\
			\partial B_{(k,p-k)}
			&= d^\dagger B_{(k-1,p-k+1)},
		\end{align}
		as 
		\begin{equation}
			\partial (b^{(k)} \smallfrown m_{(k+l)}) 
			= (\delta b^{(k)}) \smallfrown m_{(k+l)} + (-1)^k b^{(k)} \smallfrown \partial m_{(k+l)}
		\end{equation}
		and 
		\begin{equation}
			\delta b^{(q+1,-1)} = 0.
		\end{equation}
		Then, we have
		\begin{align}
			d^\dagger_0 B_{(p,0)} 
			&= (-1)^{p(q+1)}  \, b^{(q+1,-1)} \smallfrown d^\dagger_0 M_{(n,0)} \nonumber \\
			&= (-1)^{p(q+1)}  \, b^{(q+1,-1)} \smallfrown m_{(n,-1)} \, ,
		\end{align}
		which is a $p$-cycle since $\delta b^{(q+1,-1)} = 0$ and $\partial  m_{(n,-1)} = 0$. Now, to make the connection with the component $B_{(0,p)}$, let us first remark that
		\begin{align}
			\label{helpfulrelation}
			\partial B_{(1,p-1)} 
			&= b^{(q+1,-1)} \smallfrown \partial M_{(q+2,p-1)} \nonumber \\
			&= b^{(q+1,-1)} \smallfrown d^\dagger M_{(q+1,p)} \nonumber \\
			\partial B_{(1,p-1)}
			&= d^\dagger (b^{(q+1,-1)} \smallfrown M_{(q+1,p)}) \, .
		\end{align}
		This straightforwardly yields $\partial B_{(1,p-1)} = d^\dagger B_{(0,p)}$. Finally, a simple algebraic juggle yields
		\begin{equation}
			\partial B_{(0,p)} = \oint_M F(B^{[q]}) \wedge \bullet \, ,
		\end{equation}
		with an obvious meaning for the right-hand side of this equality. Hence, $\partial B_{(0,p)} \in \Omega^\mathbb{Z}_{p}(M)$ as required. It is then not difficult to see that a gauge field transformation yields a dual gauge current transformation and hence that the above construction passes to the classes as expected.      
	\end{proof}
	
	\subsection{Dual quantum fields versus quantum fields}
	
	Let us now state the first important theorem of this section.
	
	\begin{theorem}
		Through quantum evaluation, $H_D^p(M)^\star$ can be canonically identified with $D_p^H(M)$, dual gauge $p$-currents thus being representatives of the elements of $H_D^p(M)^\star$.
	\end{theorem}
	
	\begin{proof}
		This can be derived by applying the 5-lemma on exact sequence \eqref{sequencePont1} and exact sequence \eqref{Bob} below. The use of \eqref{iso1} is required when applying the 5-lemma.
	\end{proof}
	
	\begin{property}
		The $\mathbb{Z}$-module $D_p^H(M)$ sits in the following exact sequence
		\begin{equation}
			\label{Bob}
			0 \rightarrow \frac{\Omega_{p+1}(M)}{\Omega_{p+1}^\mathbb{Z}(M)} \xrightarrow{\bar{\mu}} D^H_p(M) \xrightarrow{cl} H_p(M) \rightarrow 0 \, .
		\end{equation}
	\end{property}
	
	\begin{proof}
		Just like in the case of dual quantum fields, we can associate to each dual gauge $p$-current a \v{C}ech $p$-cycle. Moreover, as by construction $d^\dagger_0 G_{(p+1,0)}$ is a \v{C}ech $(p+1)$-chain, the $p$-cycle associated with a dual gauge current transformation is a \v{C}ech boundary. Thus, we can associate to each dual quantum $p$-current an element of $H_p(M)$. This association defines the morphism
		\begin{equation}
			cl: D^H_p(M) \rightarrow H_p(M) \, ,
		\end{equation}
		which is the natural extension of the morphism $cl$ appearing in \eqref{dualfieldsequences}. Now, let $c_{(p,-1)}$ be a \v{C}ech $p$-cycle of $\mathcal{U}_M$ and let $C_{(p,0)}$ be a collection of $0$-currents with local compact support such that $c_{(p,-1)} = d_0^\dagger C_{(p,0)}$. Then, $d_0^\dagger  \partial C_{(p,0)} = \partial c_{(p,-1)} = 0$ and hence, there exists a collection $C_{(p-1,1)}$ of $1$-currents with local compact supports such that $\partial C_{(p,0)} = d^\dagger C_{(p-1,1)}$. On its turn, the collection $C_{(p-1,1)}$ fulfills $d^\dagger \partial C_{(p-1,1)} = 0$, so that $\partial C_{(p-1,1)} = d^\dagger C_{(p-2,2)}$. By repeating this procedure, we end up with a collection $C_{(0,p)}$ of $p$-currents with local compact supports. By construction, the $(p+2)$-tuple $\left( C_{(0,p)} , \dots , C_{(p,0)} ,  c_{(p,-1)} \right)$ is a dual gauge $p$-current. This shows that $cl$ is surjective.
		
		The morphism $\bar{\mu}$ is defined as follows. First, let $\mu{_0}{^0}$ be a partition of unity subordinate to $\mathcal{U}_M$. To any $(p+1)$-current $\chi$, we associate the $(p+2)$-tuple $\mu \chi = \left( d^\dagger (\mu{_0}{^0} \chi) , 0 , \dots , 0 \right)$ which is obviously a dual gauge $p$-current. This defines a morphism $\mu$ between $\Omega_{p+1}(M)$ and the group of dual gauge $(p+1)$-currents of $\mathcal{U}_M$. If $\chi \in \Omega_{p+1}^\mathbb{Z}(M)$, then
		\begin{align}
			d^\dagger (\mu{_0}{^0}\wedge\chi)
			=& (-1)^{q+1} d (\mu{_0}{^0} \wedge \chi) \nonumber \\
			=& (-1)^{q+1} (d \mu{_0}{^0}) \wedge \chi \nonumber \\
			=& (-1)^{q+1} (\hat{\partial} \mu{_1}{^1}) \wedge \chi \nonumber \\
			=& (-1)^{q+1} (\partial \mu{_1}{^1}) \wedge \chi \nonumber \\
			=& (-1)^q \partial (\mu{_1}{^1} \wedge \chi) \nonumber \\
			d^\dagger (\mu{_0}{^0}\wedge\chi)
			=& - \hat{\partial} (\mu{_1}{^1} \wedge \chi).    
		\end{align}
		Thus, $\mu \chi$ is equivalent to the dual gauge $p$-current $\left( 0 , d^\dagger (\mu{_1}{^1} \wedge \chi) , 0 , \dots , 0 \right)$. By repeating this procedure, we eventually obtain the dual gauge $p$-current $\left( 0 , \dots , \hat{\partial} (\mu{_{p+1}}{^{p+1}} \chi) , 0  \right)$, which is equivalent to $\mu \chi$. Let us show that the collection $d_0^\dagger (\mu{_{p+1}}{^{p+1}} \chi)$ is a \v{C}ech chain. First, this is a collection of $0$-currents with local compact support. Hence, if we evaluate this collection on the collection $1^{(p+1,0)}$ whose elements are the constant functions $1$ in each $U_{\alpha_0 \dots \alpha_{p+1}}$, we obtain
		\begin{equation}
			d_0^\dagger (\mu{_{p+1}}{^{p+1}} \chi) [1^{(p+1,0)}] = \chi [\mu{_{p+1}}{^{p+1}}] \, .
		\end{equation}
		Now, if the collections $\chi{^k}{_{p+k}}$ are the components of a \v{C}ech-de Rham descent of $\chi$, with $\chi^{(p+1,-1)}$ a \v{C}ech $(p+1)$-cocycle ending this descent, then
		\begin{align}
			\nonumber
			\chi [\mu{_{p+1}}{^{p+1}}] 
			=& d^\dagger \chi{^0}{_{p}} [\mu{_{p+1}}{^{p+1}}]\\
			\nonumber
			=& \chi{^0}{_{p}} [d \mu{_{p+1}}{^{p+1}}]\\
			\nonumber
			=& \chi{^0}{_{p}} [\hat{\partial} \mu{_{p+2}}{^{p+2}}]\\
			\nonumber
			=& (\hat{\delta} \chi{^0}{_{p}}) [\mu{_{p+2}}{^{p+2}}]\\
			\nonumber
			=& \dots\\
			\chi [\mu{_{p+1}}{^{p+1}}] 
			=& d^\dagger_{-1} \chi^{(p+1,-1)} [\mu{_{n}}{^{n}}]. 
		\end{align}
		As $d^\dagger_{-1} \chi^{(p+1,-1)} [\mu{_{n}}{^{n}}]$ is a collection of integers, $d_0^\dagger (\mu{_{p+1}}{^{p+1}} \chi)$ is a \v{C}ech chain. With this result, we conclude that the morphism $\mu$ extends to a morphism 
		\begin{equation}
			\bar{\mu}: \Omega_{p+1}(M)/\Omega_{p+1}^\mathbb{Z}(M) \rightarrow D^H_p(M) \, .
		\end{equation}
		This morphism reduces to the morphism $\bar{\mu}$ appearing in \eqref{dualfieldsequences} when dealing with dual gauge fields instead of dual gauge currents.
		
		It is quite obvious that $cl \circ \bar{\mu} = 0$ so that $\Ima \bar{\mu} \subset \ker cl$. To show the reverse inclusion, let us consider a dual quantum current $\boldsymbol{C}_{\{p\}} \in \ker cl$. Then, the last component of a dual gauge current $C_{\{p\}}$ representing $\boldsymbol{C}_{\{p\}}$ is a \v{C}ech boundary $\partial g_{(p+1,-1)}$. It is possible to find a collection $G_{(p+1,0)}$ such that $d_0^\dagger G_{(p+1,0)} = g_{(p+1,-1)}$ so that $C_{\{p\}}$ has the form $\left(C_{(0,p)} , \dots , C_{(p,0)} , \partial d_0^\dagger G_{(p+1,0)} \right)$. This dual gauge current is then equivalent to $\left(C_{(0,p)} , \dots , C_{(p,0)} - \partial G_{(p+1,0)} , 0 \right)$. Then, $d^\dagger (C_{(p,0)} - \partial G_{(p+1,0)}) = 0$ by the descent equations satisfied by a dual gauge current, which infers that $C_{(p,0)} - \partial G_{(p+1,0)} = d^\dagger G_{(p,1)}$. Thus, $C_{\{p\}}$ is also equivalent to $\left(C_{(0,p)} , \dots , C_{(p-1,1)} - \partial G_{(p,1)} , 0 , 0 \right)$. By repeating this procedure, we generate a dual gauge current of the form $\left(\tilde{C}_{(0,p)} , 0 , \dots , 0 \right)$ which represents $\boldsymbol{C}_{\{p\}}$. The descent equations imply that $d^\dagger \tilde{C}_{(0,p)} = 0$ and thus that $\tilde{C}_{(0,p)} = d^\dagger G_{(0,p+1)}$. Furthermore, we have
		\begin{equation}
			\partial \left( \mu{_0}{^0} \partial G_{(0,p+1)} - G_{(0,p+1)} \right) =  0 \, ,
		\end{equation}
		which implies that $\left( \mu{_0}{^0} \partial G_{(0,p+1)} -  G_{(0,p+1)} \right) = \partial H_{(1,p+1)}$ and thus that
		\begin{equation}
			\tilde{C}_{(0,p)} = d^\dagger (\mu{_0}{^0} \partial G_{(0,p+1)}) + \partial d^\dagger H_{(1,p+1)} \, .
		\end{equation}
		This implies that $\boldsymbol{C}_{\{p\}}$ admits as a representative the dual gauge $q$-current $\mu \partial G_{(0,p+1)}$ and thus that $\boldsymbol{C}_{\{p\}} = \bar{\mu} (\partial G_{(0,p+1)})$. This completes the proof that $\ker cl \subset \Ima \bar{\mu}$. So, we have the exact sequence
		\begin{equation}
			\frac{\Omega_{p+1}(M)}{\Omega_{p+1}^\mathbb{Z}(M)} \xrightarrow{\bar{\mu}} D^H_p(M) \xrightarrow{cl} H_p(M) \rightarrow 0 \, .
		\end{equation}
		Finally, let us show that $\bar{\mu}$ is injective. Let $\chi$ be a de Rham $(p+1)$-current such that $\bar{\mu}(\chi) = 0$. Then, the dual gauge current $\mu \chi$ is a dual gauge current transformation and there exists a collection $G_{(1,p)}$ of currents with local compact support\textcolor{red}{s?} such that $d^\dagger (\mu{_0}{^0} \chi) = \partial G_{(1,p)}$. Then, $\partial d^\dagger (\mu{_0}{^0} \chi) =  d^\dagger (\partial\mu{_0}{^0} \chi) = d^\dagger\chi = 0$, so that $\chi$ is $d^\dagger$-closed. According to the computation made above, we have $d^\dagger (\mu{_0}{^0}\chi) = - \partial (\mu{_1}{^1}\wedge\chi)$, so that $\mu \chi$ is equivalent to $\left(0, d^\dagger (\mu{_1}{^1} \wedge \chi) , 0 , \dots , 0 \right)$. This starts a descent which ends with the dual gauge current $\left( 0 , \dots , 0 , \pm \partial (\mu{_{p+1}}{^{p+1}} \wedge \chi) , 0 \right)$. This dual gauge current is a gauge transformation if and only if $d_0^\dagger (\mu{_{p+1}}{^{p+1}} \wedge \chi)$ is a \v{C}ech $(p+1)$-chain. Before showing this, let us point out that  $d_0^\dagger (\mu{_{p+1}}{^{p+1}} \wedge \chi)$ is a cycle since the last component of the above dual gauge current is zero. Let $u^{(p+1,-1)}$ be a \v{C}ech $(p+1)$-cocycle of $\mathcal{U}_M$, then
		\begin{equation}
			u^{(p+1,-1)} [ d_0^\dagger (\mu{_{p+1}}{^{p+1}} \wedge \chi) ] \in \mathbb{Z} \, .
		\end{equation}
		Furthermore, we have
		\begin{align}
			\nonumber
			u^{(p+1,-1)} [ d_0^\dagger (\mu{_{p+1}}{^{p+1}} \wedge \chi) ] 
			&= (d_{-1} u^{(p+1,-1)}) [ \mu{_{p+1}}{^{p+1}} \wedge \chi ] \\
			&= \chi [\mu{_{p+1}}{^{p+1}} \Wc d_{-1} u^{(p+1,-1)}] \, . 
		\end{align}
		As $\chi$ is closed, we can apply to it a \v{C}ech-de Rham descent, the components of which are denoted by $R{^{k-1}}{_{p+k+1}}$, the final component being a(n a priori real) \v{C}ech cocycle $r^{(q,-1)}$. Then, we have
		\begin{align}
			\nonumber
			\chi [\mu{_{p+1}}{^{p+1}} \Wc d_{-1} u^{(p+1,-1)}] 
			&= (d^\dagger R{^{0}}{_{p+2}}) [\mu{_{p+1}}{^{p+1}} \Wc d_{-1} u^{(p+1,-1)}]\\ 
			&= R{^{0}}{_{p+2}} [(d \mu{_{p+1}}{^{p+1}}) \Wc d_{-1} u^{(p+1,-1)}] \nonumber\\ 
			&= R{^{0}}{_{p+2}} [(\partial \mu{_{p+2}}{^{p+2}}) \Wc d_{-1} u^{(p+1,-1)}] \nonumber\\ 
			&= R{^{0}}{_{p+2}} [\partial (\mu{_{p+2}}{^{p+2}} \Wc d_{-1} u^{(p+1,-1)})] \nonumber\\ 
			&= (\delta R{^{0}}{_{p+2}}) [\mu{_{p+2}}{^{p+2}} \Wc d_{-1} u^{(p+1,-1)}] \nonumber\\
			\chi [\mu{_{p+1}}{^{p+1}} \Wc d_{-1} u^{(p+1,-1)}]
			&= (d^\dagger R{^{1}}{_{p+3}}) [\mu{_{p+2}}{^{p+2}} \Wc d_{-1} u^{(p+1,-1)}] \, ,
		\end{align}
		where we used the fact that $\delta u^{(p+1,-1)} = 0$. Repeating the procedure, we end with
		\begin{align}
			\nonumber
			\chi [\mu{_{p+1}}{^{p+1}} \Wc d_{-1} u^{(p+1,-1)}] 
			&= ( d_{-1}^\dagger r^{(q,-1)}) [\mu{_{n}}{^{n}} \Wc d_{-1} u^{(p+1,-1)}]\\ 
			&=  r^{(q,-1)} [(i_n \mu{_{n}}{^{n}}) \smallfrown u^{(p+1,-1)}] \, .
		\end{align}
		Now, as by construction $i_n \mu{_{n}}{^{n}} = m_{(n,-1)}$ is a \v{C}ech $n$-cycle of $\mathcal{U}_M$, we conclude that
		\begin{equation}
			r^{(q,-1)} [(i_n \mu{_{n}}{^{n}}) \smallfrown u^{(p+1,-1)}] \in \mathbb{Z} \, ,
		\end{equation}
		for any \v{C}ech $(p+1)$-cocycle $u^{(p+1,-1)}$ which implies that $r^{(q,-1)}$ is a \v{C}ech $q$-cocycle of $\mathcal{U}_M$, and thus that $\chi \in \Omega^\mathbb{Z}_{p+1}(M)$. Accordingly, $\bar{\mu}$ is injective.
	\end{proof}
	
	Let us note that for any gauge $p$-field $A^{[p]}$ we have
	\begin{align}
		(\mu \chi) [A^{[p]}] 
		&= d^\dagger (\mu{_0}{^0} \chi) [A^{[(0,p)}] \nonumber \\
		&= (\mu{_0}{^0} \chi) [d A^{[(0,p)}]\nonumber\\
		&= (\mu{_0}{^0} \chi) [\delta_{-1} F(A)]\nonumber\\
		(\mu \chi) [A^{[p]}]	    
		&= \chi [F(A)] \in \mathbb{Z} \, .
	\end{align}
	
	\begin{exercise}
		Show that $D^H_p(M)$ sits in the following exact sequence
		\begin{equation}
			0 \rightarrow H^{q}(M,\mathbb{R}/\mathbb{Z}) \xrightarrow{i} D^H_p(M) \xrightarrow{\bar{\partial}} \left( \frac{\Omega^p(M)}{\Omega_\mathbb{Z}^p(M)} \right)^\star \rightarrow 0 \, ,
		\end{equation}
		where the morphisms $i$ and $\bar{\partial}$ are the natural extensions of the morphisms $i$  and $\bar{\partial}$ appearing in \eqref{dualfieldsequences}.
	\end{exercise}
	
	Now that we have shown how to identify the elements of $H_D^p(M)^\star$ as dual quantum $p$-currents of $M$, let us relate this $\mathbb{Z}$-module with $H^D_p(M)$.
	
	\begin{theorem}
		A partition of unity subordinate to $\mathcal{U}_M$ defines a homomorphism between the group of gauge $p$-currents and the group of dual gauge $p$-currents. This homomorphism becomes an isomorphism between the $\mathbb{Z}$-modules $H^D_p(M)$ and $H_D^p(M)^\star$ which turns out to be independent of the partition of unity chosen.
	\end{theorem}
	
	\begin{proof}
		If $B_{[p]}$ is a gauge $p$-current, we just have to set
		\begin{equation}
			\left\{ 
			\begin{gathered}
				C^{B,\mu}_{(0,p)} = d^\dagger \Big(\sum_{k=0}^{q} (-1)^k  \mu_k^k \Wc B^k_{n-q+k}\Big) + (b^{(q+1,-1)} \smallfrown \mu_{q+1}^{q+1}) \hfill \\ 
				C^{B,\mu}_{(k,p-k)} = (-1)^{k(q+1)} \, b^{(q+1,-1)} \smallfrown \mu_{q+k+1}^{q+k+1} \hfill \\
				c^{B,\mu}_{(p,-1)} =  b^{(q+1,-1)} \smallfrown m_{(n,-1)} \hfill 
			\end{gathered} \right. \, ,
		\end{equation}
		with $p$ and $q$ quantum dual integers. Then, $C^{B,\mu}_{\{p\}} = \left( C^{B,\mu}_{(0,p)} , \hdots , C^{B,\mu}_{(p,0)} , c^{B,\mu}_{(p,-1)} \right)$ is a dual gauge $p$-current. It is not hard to check that another partition of unity will produce an equivalent dual gauge $p$-current. Finally, by construction, for any gauge $p$-field $A^{[p]}$, we have
		\begin{equation}
			C^{B,\mu}_{\{p\}} [A^{[p]}] = \sum\limits_{k = 0}^p (-1)^k C^{B,\mu}_{(k,p-k)} \left[ A^{(k,p-k)} \right] \stackrel{\mathbb{R}/\mathbb{Z}}{=} (B_{[p]} \star A^{[p]}) \left[ \mu^{\{-1\}} \right] \, ,
		\end{equation}
		and thus at the level of classes
		\begin{equation}
			\label{trueevalAstarB}
			\boldsymbol{C}^B_{\{p\}} \eval{\boldsymbol{A}^{[p]}} \stackrel{\mathbb{R}/\mathbb{Z}}{=} (\boldsymbol{B}_{[p]} \star \boldsymbol{A}^{[p]}) \eval{\boldsymbol{1}} = \boldsymbol{B}_{[p]} \eval{\boldsymbol{A}^{[p]} \DBc \boldsymbol{1}} \, .
		\end{equation}
		
	\end{proof}
	
	Now that we have representatives for the dual quantum currents (or equivalently for the elements of $H_D^p(M)^\star$) and a quantum integration, i.e., an evaluation formula defined with the help of these representatives, it is possible to give a meaning to the evaluation $(B_{[p]} \star A_{[q]}) [\mu^{\{-1\}}]$ just like it is possible to give a meaning to $(j_1 \wedge j_2)[1]$ for some particular de Rham currents $j_1$ and $j_2$ \cite{dR55}. Thanks to the use of a partition of unity, the only possibly ill-defined quantities which might appear in defining $(B_{[p]} \star A_{[q]}) [\mu^{\{-1\}}]$ are due to the product of de Rham currents of $B_{[p]}$ and $A_{[q]}$. The situation would be quite hopeless if we were trying to use the evaluation formula based on a decomposition of $M$ since, then, more de Rham currents products would be involved. Nevertheless, even if a partition of unity is used, it remains true that $(B_{[p]} \star A_{[q]}) [\mu^{\{-1\}}]$ is, in general, ill-defined. However, in the case where an evaluation like $(B_{[p]} \star A_{[q]}) [\mu^{\{-1\}}]$ has a meaning, it is natural to consider it as defining $(\boldsymbol{B}_{[p]} \star \boldsymbol{A}_{[q]}) \eval{\boldsymbol{1}}$.
	
	\begin{definition}
		We say that $(\boldsymbol{B}_{[p]} \star \boldsymbol{A}_{[q]}) \eval{\boldsymbol{1}}$ is well-defined if there exist some gauge currents $B_{[p]}$ and $A_{[q]}$ representing $\boldsymbol{B}_{[p]}$ and $\boldsymbol{A}_{[q]}$, respectively, for which
		$(B_{[p]} \star A_{[q]}) \left[\mu^{\{-1\}} \right]$ has a meaning.
	\end{definition}
	
	It is not hard to check that if $\boldsymbol{A}_{[q]}$ and $\boldsymbol{B}_{[p]}$ both derive from two quantum fields $\boldsymbol{A}^{[p]}$ and $\boldsymbol{B}^{[q]}$, then
	\begin{equation}
		(\boldsymbol{B}_{[p]} \star \boldsymbol{A}_{[q]}) \eval{\boldsymbol{1}} = \oint_M \boldsymbol{B}^{[q]} \star \boldsymbol{A}^{[p]} = \boldsymbol{B}_{\{p\}} \eval{\boldsymbol{A}^{[p]}}\, .
	\end{equation}
	Furthermore, if we set
	\begin{equation}
		\oint_M \boldsymbol{B}_{[p]} \star \boldsymbol{A}_{[q]} = (\boldsymbol{B}_{[p]} \star \boldsymbol{A}_{[q]}) \eval{\boldsymbol{1}} \, ,
	\end{equation}
	then the above definition allows to consider the formal generalized $\U1$ BF quantum action
	\begin{equation}
		S_{BF,k} (\boldsymbol{A}_{[q]},\boldsymbol{B}_{[p]}) = k \oint_M \boldsymbol{B}_{[p]} \star \boldsymbol{A}_{[q]} \, .
	\end{equation}
	
	\vspace{0.5cm}
	
	We can now state the following important consequence of the construction.
	
	\begin{property}
		We can associate with any $p$-cycle $z$ of $M$ a quantum $p$-current $\boldsymbol{Z}_{[p]}$ such that
		\begin{equation}
			\oint_z \boldsymbol{A}^{[p]} \stackrel{\mathbb{R}/\mathbb{Z}}{=} (\boldsymbol{Z}_{[p]} \star \boldsymbol{A}^{[p]}) \eval{\boldsymbol{1}} \, .
		\end{equation}
	\end{property}
	
	\noindent How to explicitly construct a gauge $p$-current representing $\boldsymbol{Z}_{[p]}$ is explained in detail in \cite{BGST05}.
	
	Let us end this subsection with the following remark. The Pontryagin dual of $H^D_p(M)$ can be canonically identified with $D_H^p(M)$, the $\mathbb{Z}$-module of dual quantum $q$-fields of $M$. Moreover, $H^D_p(M)^\star \simeq H_D^p(M)$.
	
	\section{Conclusion}
	
	In this first article we focused on studying the various mathematical objects which will be used in the construction of the generalized $\U1$ BF theory. The main points were to extend the notion of quantum fields, which are the elements of our original quantum configuration space, in order to obtain singular fields. The reason for this extension is twofold. Firstly, we know that singular fields are necessary in any quantum field theory, whether it is considered in the constructive approach or in the more formal approach of Feynman. Secondly, thanks to this, it is possible to see Wilson loops as fields. More precisely, to any $p$-cycle which generates a Wilson loop observable, we can associate a quantum current. This will be very convenient when dealing with the expectation values of such observables in the generalized $\U1$ BF theory. Let us point out that what we have done here, having in mind the generalized $\U1$ BF theory, can be done in the context of the generalized $\U1$ Chern-Simons theory. However, while we can consider the BF case in any dimension and with fields of any degrees, as long as these degrees are quantum complement with each other, in the CS case we must consider $(4l+3)$-dimensional closed manifolds as well as quantum $(2l+1)$-currents only. This was discussed in detail in \cite{GPT13} from the Deligne-Beilinson point of view but also from the point of view of standard quantum field theory in $\mathbb{R}^{4l+3}$.
	
	Latter in this series of articles, we will specifically deal with the partition function and expectation values of observables. A generalized TV construction will also be presented and the induced manifold invariant will be compared with the partition function. Finally, the Drinfeld construction will be exhibited and, as in the $3$-dimensional case, a discrete BF theory will naturally appear.
	
	\paragraph{Acknowledgements}  P.M. thanks the University of Notre Dame, Indiana, USA, at which he started this work, and his support from the NSF grant 1947155 and the JTF grant 61521. He also acknowledges partial support of NSF grant 200020-192080 of the Simons Collaboration on Global Categorical Symmetries, and of the COST Action 21109 - Cartan geometry, Lie, Integrable Systems, quantum group Theories for Applications (CaLISTA). This research was partly supported by the NCCR SwissMAP, funded by the Swiss National Science Foundation.
	
	
	
	
	
	
	\vfill\eject

\end{document}